\newif\iflong
\newif\ifhidecomments
\let\accentvec\vec
\documentclass[10pt]{article}

\let\vec\accentvec
\usepackage{fullpage}
\usepackage{ifthen}
\usepackage{xspace}
\usepackage{amsmath, amssymb, amsfonts}
\usepackage{amsthm}
\usepackage[english]{babel}
\usepackage{verbatim}
\usepackage{subfig}

\allowdisplaybreaks

\theoremstyle{plain}
\newtheorem{theorem}{Theorem}

\newtheorem{corollary}[theorem]{Corollary}
\newtheorem{proposition}[theorem]{Proposition}
\newtheorem{lemma}[theorem]{Lemma}

\theoremstyle{definition}
\newtheorem{definition}{Definition}
\newtheorem{example}{Example}

\theoremstyle{remark}
\newtheorem{remark}{Remark}

% ...
% ...
\providecommand{\SET}[1]{\ensuremath{\left\{ #1 \right\}}\xspace}
\providecommand{\Set}[2]{\ensuremath{\SET{#1 : #2}}\xspace}

\providecommand{\Expect}[2][]{\ensuremath{%
\ifthenelse{\equal{#1}{}}{\mathbf{E}}{\mathbf{E}_{#1}}%
\left[#2\right]}\xspace}

\newcommand{\altvec}{\ensuremath{\mathbf{\alpha}}\xspace}
\newcommand{\PoA}{\ensuremath{\text{PoA}}\xspace}
\newcommand{\PoS}{\ensuremath{\text{PoS}}\xspace}
\newcommand{\RPoA}{\ensuremath{\text{RPoA}}\xspace}
\newcommand{\pne}{\ensuremath{\text{PNE}}\xspace}
\newcommand{\mne}{\ensuremath{\text{MNE}}\xspace}
\newcommand{\ce}{\ensuremath{\text{CE}}\xspace}
\newcommand{\cce}{\ensuremath{\text{CCE}}\xspace}
\newcommand{\maxbeta}{\ensuremath{\hat\alpha}}
\newcommand{\minbeta}{\ensuremath{\check\alpha}}
\renewcommand{\varnothing}{\ensuremath{\emptyset}}

\newcommand{\Omit}[1]{}
\renewcommand{\enspace}{}

\title{The Robust Price of Anarchy of Altruistic Games}

\author{%
  Po-An Chen\thanks{%
    Department of Computer Science, University of Southern California, USA.
    Work done in part while visiting CWI Amsterdam.
    Email: poanchen@usc.edu.
  }
  \and
  Bart de Keijzer\thanks{%
    Algorithms, Combinatorics and Optimization, CWI Amsterdam, The Netherlands.
    Email: b.de.keijzer@cwi.nl.
  }
  \and
  David Kempe\thanks{%
    Department of Computer Science, University of Southern California, USA.
    Email: dkempe@usc.edu.
  }
  \and
  Guido Sch\"{a}fer\thanks{%
    Algorithms, Combinatorics and Optimization, CWI Amsterdam and
    Department for Econometrics and Operations Research, VU University Amsterdam, The Netherlands.
    Email: g.schaefer@cwi.nl.
  }
}
\date{}

\begin{document}
\sloppy
\maketitle
\begin{abstract}
We study the inefficiency of equilibria for various classes of games
when players are (partially) altruistic. We model altruistic behavior
by assuming that player $i$'s perceived cost is a convex combination
of $1-\alpha_i$ times his direct cost and $\alpha_i$ times the social
cost. Tuning the parameters $\alpha_i$ allows
smooth interpolation between purely selfish and purely altruistic
behavior.
Within this framework, we study altruistic extensions of linear
congestion games, fair cost-sharing games and valid utility games.

We derive (tight) bounds on the price of anarchy of these games for
several solution concepts. Thereto, we suitably adapt the
\emph{smoothness} notion introduced by Roughgarden
% I suggest just not citing a paper.
%\gscomment{not sure how we should deal with the citation here differently}
%[\emph{Intrinsic robustness of the price of anarchy}, STOC '09]
%\dkcomment{I would not put citations in an abstract}
and show that it captures the essential properties to determine
the \emph{robust price of anarchy} of these games. Our bounds show that
%\pcedit{interesting and counter-intuitive} trend:
for congestion games and cost-sharing games,
the worst-case robust price of anarchy increases with increasing altruism,
while for valid utility games, it remains constant and is not affected by altruism.
However, the increase in the price of anarchy is not a universal phenomenon: for
symmetric singleton linear congestion games, we derive a bound on the
pure price of anarchy that decreases as the level
of altruism increases. Since the bound is also strictly lower than
the robust price of anarchy, it exhibits a natural example in which
Nash equilibria are more efficient than more permissive notions of equilibrium.
%\remove{We also show that the price of stability decreases with increasing altruism.}
\end{abstract}

\thispagestyle{empty}
\newpage
\pagestyle{plain}
\setcounter{page}{1}

\section{Introduction}
Many large-scale decentralized systems, such as infrastructure
investments or traffic on roads or computer networks,
bring together large numbers of individuals with different and
oftentimes competing objectives.
When these individuals choose actions to benefit themselves, the
result is frequently suboptimal for society as a whole.
This basic insight has led to a study of such systems from the
viewpoint of game theory, focusing on the inefficiency of stable
outcomes. Traditionally, ``stable outcomes'' have been associated with
pure Nash equilibria of the corresponding game. The notions of
\emph{price of anarchy} \cite{koutsoupias2} and \emph{price of
  stability} \cite{anshelevich} provide natural measures of the system
degradation, by capturing the
degradation of the worst and best Nash equilibria, respectively, compared
to the socially optimal outcome.

However, the predictive power of such bounds has been questioned on
(at least) two grounds:
\begin{enumerate}\itemsep0pt
\item The adoption of Nash equilibria as a prescriptive solution
concept implicitly assumes that players are able to reach such
equilibria. In particular in light of several known hardness
results for finding Nash equilibria, this assumption is very suspect
for computationally bounded players.
In response, recent work has begun analyzing the outcomes of natural
response dynamics \cite{blum:even-dar,blum,roughgarden}, as well as
more permissive solution concepts such as correlated or coarse
correlated equilibria \cite{aumann,hannan,roughgarden3}.
This general direction of inquiry has become known as ``robust price
of anarchy''.
\item The assumption that players seek only to maximize
their own utility is at odds with altruistic behavior
routinely observed in the real world.
While modeling human incentives and behavior accurately is a
formidable task, several papers have proposed natural models of
altruism \cite{ledyard,levine} and analyzed its impact on the outcomes
of games \cite{caragiannis,chen:david,chenkempe,networkdesign}.
\end{enumerate}

The goal of this paper is to begin a thorough investigation of the
effects of relaxing both of the standard assumptions
simultaneously, i.e., considering the combination of weaker
solution concepts and notions of partially altruistic behavior by players.
In Section \ref{sec:preliminaries}, we formally define the
\emph{altruistic extension} of
an $n$-player game in the spirit of past work on altruism
(see \cite[p.~154]{ledyard} and \cite{skopalik,chen:david}):
player $i$ has an associated altruism parameter $\alpha_i$,
and player $i$'s cost (or payoff) is a convex combination of $(1-\alpha_i)$ times his direct cost (or payoff)
and $\alpha_i$ times the social cost (or social welfare). By tuning the
parameters $\alpha_i$, this model allows smooth interpolation
between pure selfishness ($\alpha_i = 0$) and pure altruism ($\alpha_i=1$).

In order to analyze the degradation of system performance in light
of partially altruistic behavior,
we extend the notion of \emph{robust price of anarchy}
\cite{roughgarden} to altruistic extensions, and show that a suitably
adapted notion of \emph{smoothness} \cite{roughgarden} captures the
properties of a system that determine its robust price of anarchy.
%\dkdelete{Such bounds are not a trivial consequence of Roughgarden's work \cite{roughgarden}: the social cost function of the game remains unchanged, and is thus not an upper bound on the sum of utilities of the modified game with altruism.}

We use this framework to analyze three classes of games:
\begin{enumerate}\itemsep0pt
\item In a \emph{cost-sharing game} \cite{anshelevich},
players choose subsets of resources, and all players choosing the
same resource share its cost evenly.
Thus, cost-sharing games model scenarios in which individual
players have an interest in building infrastructure, and can
share the cost of infrastructure that benefits several of the players.
Using our framework, we derive a bound of $n/(1-\maxbeta)$ on the
robust price of anarchy of these games, where $\maxbeta$ is the
maximum altruism level of a player. This bound is tight
for uniformly altruistic players.

\item In \emph{utility games} \cite{vetta}, players
choose subsets of resources and derive utility of the chosen set.
The total welfare is determined by a submodular function of the union
of all chosen sets.
Utility games thus model scenarios in which different players
build infrastructure with different objectives, and the lack of
coordination may be societally suboptimal.
We derive a bound of 2 on the robust price of anarchy of these
games. In particular, the bound remains at 2 regardless of the
(possibly different) altruism levels of the players.
This bound is tight.

\item We revisit and extend the analysis of
\emph{atomic congestion games} \cite{roughgarden4}, in which
players choose subsets of resources whose costs increase (linearly)
with the number of players using them.
%We study both the general case and the special case of singleton congestion games.
%and as resources are chosen by many players, their cost increases (linearly), to all players using them.
Thus, they are natural models of traffic on roads or in
computer networks as well as scheduling on machines,
where selfish choices can lead to overcongestion of
resources which would be much faster if used in moderation.
Caragiannis et al.~\cite{caragiannis} recently derived a tight bound of
$(5+4\alpha)/(2+\alpha)$ on the pure price of anarchy when all
players have the \emph{same} altruism level $\alpha$.\footnote{The altruism
    model of \cite{caragiannis} differs from ours in a slight
    technicality discussed in Section \ref{sec:preliminaries} (Remark~\ref{rem:rel-Car}).
    Therefore, various bounds we cite here are stated differently in
    \cite{caragiannis}.}
%Our first observation is that\pcedit{, with \emph{uniform} altruism level
%  $\alpha$,} the bound of $(5+4\alpha)/(2+\alpha)$ on the pure
%price of anarchy recently derived by Caragiannis et
%al.~\cite{caragiannis} for uniformly altruistic players in fact fits
%our smoothness
%  framework for altruistic games.
  Our framework makes it an easy observation that their proof in fact
  bounds the robust price of anarchy. We generalize their bound to the
  case when different players have different altruism levels,
  obtaining a bound in terms of the maximum and
  minimum altruism levels. This partially answers
%  It is thus also a bound on the robust price of
%  anarchy. We extend this result by deriving a bound on the robust price of anarchy
%  for \emph{non-uniform} altruism levels in terms of the maximum and minimum
%  altruism levels, partially answering
  an open question from \cite{caragiannis}.
  For the special case of symmetric singleton congestion games (which
  corresponds to selfish scheduling on machines), we extend our study
  of non-uniform altruism and obtain an improved bound of
  $(4-2\alpha)/(3-\alpha)$ on the price of anarchy when an $\alpha$-fraction of the
  players are entirely altruistic and the remaining players are entirely selfish.
\end{enumerate}

Notice that many of these bounds on the robust price of anarchy reveal a counter-intuitive trend:
at best, for utility games, the bound is independent of the
level of altruism, and for congestion games and cost-sharing
games, it actually \emph{increases} in the altruism level,
unboundedly so for cost-sharing games.
%; the only exception is the setting of symmetric singleton congestion games.
% The intuition behind the increase in the price of anarchy
% is the following: there are instances in which all players get stuck choosing
% the wrong resources. A deviation by one player affects not only him,
% but also others: for congestion games, the player may increase
% the cost on the resources he switches to, while for cost-sharing
% games, there will be fewer remaining players to share the cost of the
% player's current resource. Thus, partially altruistic players have even
% stronger disincentive to deviate from the suboptimal strategy,
% meaning that even worse system states are stable.
Intuitively, this phenomenon is explained by the fact that a change of
strategy by player $i$ may affect many players.
An altruistic player will care more about these other players than a
selfish player; hence, an altruistic player accepts more states as
``stable''. This suggests that the best stable solution can also be
chosen from a larger set, and the price of
stability should thus decrease. Our results on the price of stability
lend support to this intuition: for congestion games, we derive an
upper bound on the
price of stability which decreases as $2/(1+\alpha)$; similarly, for
cost-sharing games, we establish an upper bound which decreases as
$(1-\alpha) H_n + \alpha$.

The increase in the price of anarchy is not a universal
phenomenon, demonstrated by \emph{symmetric singleton} congestion games.
Caragiannis et al.~\cite{caragiannis} showed a bound of
$4/(3+\alpha)$ for pure Nash equilibria with uniformly altruistic
players, which decreases with the altruism level $\alpha$.
Our bound of $(4-2\alpha)/(3-\alpha)$ for mixtures of entirely altruistic and selfish players is also decreasing in the fraction of entirely altruistic players.
We also extend an example of L\"{u}cking et al.~\cite{luecking} to
show that symmetric singleton congestion games may have a mixed price
of anarchy arbitrarily close to 2 for arbitrary altruism levels. In
light of the above bounds, this establishes that pure Nash equilibria
can result in strictly lower price of anarchy than weaker solution
concepts.

\bigskip\noindent
\textbf{Related Work.~~}
Much of our analysis is based on extensions of the notion of
\emph{smoothness} as proposed by Roughgarden \cite{roughgarden} (see
Section~\ref{sec:smoothness}). The basic idea is to bound the sum of
cost increases of individual players switching strategies by a
combination of the costs of two states.
Because these types of bounds capture local improvement
dynamics, they bound the price of anarchy not only for Nash equilibria, but also
more general solution concepts, including coarse correlated equilibria.
The smoothness notion was recently refined in the \emph{local
  smoothness} framework by Roughgarden and Schoppmann \cite{roughgarden3}.
They require the types of bounds described above only for nearby
states, thus obtaining tighter bounds, albeit only for more
restrictive solution concepts and convex
strategy sets. Using the local smoothness framework,
they obtained optimal upper bounds for atomic splittable congestion games.
Nadav and Roughgarden \cite{roughgarden2} showed that smoothness
bounds apply all the way to solution concepts called ``average
coarse correlated equilibrium,'' but not beyond.

A comparison between the costs in worst-case outcomes under solution
concepts of different generality was recently undertaken by Bradonjic
et al.~\cite{bradonjic} under the name ``price of mediation:''
specifically for the case of symmetric singleton congestion
games with convex latency functions, they showed that the ratio
between the most expensive correlated equilibrium and the most
expensive Nash equilibrium can grow exponentially in the number of
players.

Hayrapetyan et al.~\cite{hayrapetyan:tardos:wexler:collusion} studied
the impact of ``collusion'' in network congestion games, where players
form coalitions to minimize their collective cost. These coalitions
are assumed to be formed exogeneously, i.e., conceptually, each
coalition is replaced by a ``super-player'' that acts on behalf of its
members.
%The authors consider both non-atomic and atomic network congestion games.
%Conceptually, each coalition is replaced by some super-player that acts selfishly on behalf of the group.
The authors show that collusion in network congestion games can lead
to Nash equilibria that are inferior to the ones of the collusion-free
game (in terms of social cost).
They also derive bounds on the the price of anarchy caused by collusion.
%Even though their studies seem reminiscent of our investigations for congestion games, the viewpoint adopted in their work differs from the one that we adopt here:
Note that the cooperation within each coalition can be interpreted as
a kind of ``locally'' altruistic behavior, i.e., each player only
cares about the cost of the members of his coalition.
%Also observe that each player can be regarded as being entirely altruistic (towards the members of his coalition).
In a sense, the setting considered in
\cite{hayrapetyan:tardos:wexler:collusion} can therefore be regarded
as being orthogonal to the viewpoint that we adopt in this paper: in
their setting, players are assumed to be entirely altruistic but
locally attached to their coalitions. In contrast, in our setting,
players may have different levels of altruism but locality does not
play a role.

Several recent studies investigate ``irrational'' player behavior in
games; examples include studies on malicious (or spiteful) behavior
\cite{babaioff:kleinberg:papadimitriou:malicious,BSS07,chenkempe,KV07}
and unpredictable (or Byzantine) behavior \cite{blum,
  MSW06,Roth08}. The work that is most related to our work in this
context is the one by Blum et al.~\cite{blum}. The authors consider
repeated games in which every player is assumed to minimize his own
regret. They derive bounds on the inefficiency, called \emph{total
  price of anarchy}, of the resulting outcomes for certain classes of
games, including congestion games and valid utility games. The
exhibited bounds exactly match the respective price of anarchy and
even continue to hold if only some of the players minimize their
regret while the others are Byzantine.
The latter result is surprising in
the context of valid utility games because it means that the
price of total anarchy remains at $2$, even if additional players
are added to the game that behave arbitrarily.
Our findings allow us to draw an even more dramatic conclusion.
Our bounds on the robust price of anarchy also extend to the total
price of anarchy of the respective repeated games
(see Section~\ref{sec:pre-res}). As a consequence, our
result for valid utility games implies that the price of total
anarchy would remain at 2, even if the ``Byzantine'' players
were to act altruistically.
That is, while the result in \cite{blum} suggests that
arbitrary behavior does not harm the inefficiency of the final
outcome, our result shows that altruistic behavior does not help.

%Babaioff et al. \cite{babaioff:kleinberg:papadimitriou:malicious} study (non-atomic) network congestion games in the presence of malicious players that seek to maximize the overall cost of the selfish players, not caring about his own cost. The authors show that malicious player behavior sometimes leads to superior outcomes and sometimes to inferior ones (in terms of social cost).

If players' altruism levels are not uniform, then even the existence
of pure Nash equilibria is not obvious. Hoefer and Skopalik established it for
several subclasses of atomic congestion games \cite{skopalik}; for the
generalization of arbitrary player-specific cost functions, Milchtaich
\cite{milchtaich} showed existence for singleton congestion games, and
Ackermann et al.~\cite{ackermann} for matroid congestion games, in
which the strategy space of each player is the basis of a matroid on
the set of resources.

\bigskip\noindent
\textbf{Models of Altruism.~~}
Models of altruism either identical or very similar to the one in this
paper have been studied in several papers. Perhaps the first published
suggestion of a similar model is due to Ledyard \cite{ledyard}, but
since then, different variations of it have been studied more
extensively, e.g., \cite{caragiannis,chen:david,chenkempe,networkdesign}.
The main difference is that in some of these models,
linear combinations (rather than convex combinations) are considered,
e.g., with the selfish term having a factor of 1. For most of these
variations, a straightforward scaling of the coefficients shows
equivalence with the model we consider here.
The altruism model can be naturally extended
to include $\alpha_i < 0$, modeling spiteful behavior (see, e.g., \cite{chenkempe}).
While the modeling extension is natural, several results in this and
other papers do not continue to hold directly for negative $\alpha_i$.
Our model is strictly more general than some of the previous work in
that the social cost function need not be the sum of all players'
costs, but rather only needs to be bounded by the sum.

Besides models based on linear combinations of individual players'
costs (as well as social welfare), several other approaches have been
studied. Generally, altruism or other ``other-regarding'' social
behavior has received some attention in the behavioral economics
literature (e.g., \cite{gintis}).
Alternative models of altruism and spite have been proposed
by Levine \cite{levine}, Rabin \cite{rabin} and Geneakoplos et
al.~\cite{geanakoplos}.
These models are designed more with the goal of modeling the
psychological processes underlying spite or altruism (and
reciprocity): they involve players forming beliefs about other
players. As a result, they are well-suited for experimental work, but
perhaps not as directly suited for the type of analysis in this paper.

\section{Preliminaries} \label{sec:preliminaries}

Let $G = (N, \{\Sigma_i\}_{i \in N}, \{C_i\}_{i \in N})$ be a finite strategic game,
where $N = [n]$ is the set of players, $\Sigma_i$ the strategy space of player $i$, and
$C_i : \Sigma \rightarrow \mathbb{R}$ the cost function of player $i$,
mapping every strategy profile
$s \in \Sigma = \Sigma_1 \times \dots \times \Sigma_n$ to the player's
direct cost.
Unless stated otherwise, we assume that every player $i$ wants to minimize his individual cost function
$C_i$. We also call such games \emph{cost-minimization games}. A \emph{social cost} function
$C: \Sigma \rightarrow \mathbb{R}$ maps strategies to social costs.
We require that $C$ is \emph{sum-bounded}, that is, $C(s) \leq \sum_{i=1}^n C_i(s)$ for all $s \in \Sigma$.
We study \emph{altruistic extensions} of strategic
games equipped with sum-bounded social cost functions.
Our definition is based on one used (among others) in
\cite{chen:david}, and similar to ones given in \cite{caragiannis,chenkempe,ledyard}.

\begin{definition}[Altruistic extension] \label{def:alt-ext}
%Let $G = (N, \{\Sigma_i\}_{i \in N}, \{C_i\}_{i \in N})$ be a cost-minimization game with a sum-bounded social cost function $C$.
Let $\altvec \in [0,1]^n$.
The \emph{$\altvec$-altruistic extension} of $G$
(or simply \emph{$\altvec$-altruistic game})
is defined as the strategic game
$G^\altvec = (N, \{\Sigma_i\}_{i \in N}, \{C_i^{\altvec}\}_{i \in N})$, where
for every $i \in N$ and $s \in \Sigma$,
\begin{eqnarray*}
C_i^{\altvec}(s) & = & (1 - \alpha_i)C_i(s) + \alpha_i C(s).
\end{eqnarray*}

\end{definition}

Thus, the perceived cost that player $i$ experiences
is a convex combination of his direct (selfish) cost and the social cost; we
call such a player \emph{$\alpha_i$-altruistic}.\footnote{%
We note that the altruistic part of an individual's perceived cost
does not recursively take other players' \emph{perceived} cost into
account. Such recursive definitions of altruistic utility have been
studied, e.g., by Bergstrom \cite{bergstrom:benevolent}, and can be
reduced to our definition under suitable technical conditions.}
When $\alpha_i =0$, player $i$ is entirely selfish;
thus, $\altvec=\mathbf{0}$ recovers the original game.
%\remove{with entirely selfish players}.
A player with $\alpha_i = 1$ is entirely altruistic.
Given an altruism vector $\alpha \in [0,1]^n$, we let
$\maxbeta = \max_{i \in N} \alpha_i$ and $\minbeta = \min_{i \in N} \alpha_i$
%$$
%\textstyle\maxbeta = \max_{i \in N} \alpha_i \quad\text{and} \quad \minbeta = \min_{i \in N} \alpha_i
%$$
denote the maximum and minimum altruism levels, respectively.
When $\alpha_i = \alpha$ (a scalar) for all $i$, we
call such games \emph{uniformly $\alpha$-altruistic games}.
%with $\alpha$ being a scalar (instead of a vector) that characterizes
%the common altruism level.

\begin{remark}\label{rem:rel-Car}
In a recent paper, Caragiannis et al.~\cite{caragiannis} model
uniformly altruistic players by defining the
perceived cost of player $i$ as $(1-\xi) C_i(s) + \xi (C(s)-C_i(s))$,
where $\xi \in [0, 1]$. It is not hard to see that in the range
  $\xi \in [0, \frac12]$ this definition is equivalent to ours
by setting $\alpha = \xi/(1-\xi)$ or
$\xi = \alpha/(1+\alpha)$.\footnote{The model of \cite{caragiannis} with
  $\xi \in (\frac12, 1]$ has players assign strictly more weight to
  others than to themselves, a possibility not present in our model
  since we consider altruism to be caring about others' costs at most
  as much as about one's own cost.}
\end{remark}

The altruistic extension of a \emph{payoff-maximization game}, in
which players seek to maximize their payoff functions $\{\Pi_i\}_{i
\in N}$, with a social welfare function $\Pi$ is defined analogously
to Definition~\ref{def:alt-ext}; the only difference is that
every player $i$ wants to \emph{maximize} $\Pi_i^{\altvec}$ instead of
minimizing $C_i^{\altvec}$ here.

\subsection{Equilibrium Concepts} \label{sec:equilibrium}
We study the inefficiency of equilibria in altruistic extensions
of various games. The most general equilibrium concept that we will
deal with is the following one.

\begin{definition}[Coarse equilibrium]
A \emph{coarse equilibrium} (or \emph{coarse correlated equilibrium})
of a game $G$ is a probability distribution $\sigma$ over
$\Sigma = \Sigma_1 \times \cdots \times \Sigma_n$ with the following property:
if $s$ is a random variable with distribution $\sigma$, then for each
player $i$, and all $s_i^* \in \Sigma_i$:
\begin{eqnarray}\label{eq:coarseeq}
\Expect[s \sim \sigma]{C_i(s)}
& \leq & \Expect[s_{-i} \sim \sigma_{-i}] {C_i(s_i^*, s_{-i})},
\end{eqnarray}
where $\sigma_{-i}$ is the
projection of $\sigma$ on $\Sigma_{-i} = \Sigma_1 \times \cdots \times
\Sigma_{i-1} \times \Sigma_{i+1} \times \cdots \times \Sigma_n$.
\end{definition}

The set of all coarse equilibria is also known as the
\emph{Hannan Set} (see, e.g., \cite{young}).
It includes several other solution concepts,
such as correlated equilibria, mixed Nash equilibria and pure Nash equilibria.
We briefly review these equilibrium notions.

Informally, the difference between a coarse equilibrium and
a \emph{correlated equilibrium} is the following: in a coarse
equilibrium, it is required that a player ``adheres'' to $s$ when he
is informed of the distribution $\sigma$ from which $s$ is drawn. In
a correlated equilibrium, a player is only required to adhere to $s$
when he is informed of the distribution $\sigma$ as well as the
strategy that has been drawn for him, i.e., that he will play under $s$.
%More formally, this means that in a correlated equilibrium,
%for all $s_i^* \in \Sigma_i$,
%\begin{eqnarray*}
%\Expect[s \sim \sigma]{C_i(s)}
%& \leq & \Expect[s \sim \sigma] {C_i(s_i^*, s_{-i})} \enspace .
%\end{eqnarray*}
%(By $(s_i^*, s_{-i})$, we mean the strategy profile obtained from
%  $s$ when we replace $s_i$ with $s_i^*$.)

A \emph{mixed Nash equilibrium} is a coarse equilibrium whose
distribution $\sigma$ is the product of \emph{independent}
distributions $\sigma_1, \ldots, \sigma_n$ for the players.
Thus, any mixed Nash equilibrium is also a correlated equilibrium.
A \textit{pure Nash equilibrium} is a strategy
profile $s$ such that for each player $i$,
$C_i(s) \leq C_i(s_i^*, s_{-i})$ for all $s_i^* \in \Sigma_i$.
A pure Nash equilibrium is a special case of a
mixed Nash equilibrium where the support of $\sigma_i$ has cardinality
$1$ for all $i$.

We use $\pne(G)$, $\mne(G)$, $\ce(G)$, and $\cce(G)$,
to denote the set of pure Nash equilibria, mixed Nash
equilibria, correlated equilibria, and coarse equilibria of a game
$G$, respectively.

The \emph{price of anarchy} \cite{koutsoupias2} and
\emph{price of stability} \cite{anshelevich} are natural ways of
quantifying the inefficiency of equilibria for classes of games:

\begin{definition}[Price of anarchy, price of stability]\label{def:poa}
Let $S\subseteq \Sigma$ be a set of strategy profiles for a cost-minimization game
$G$ with social cost function $C$, and let $s^*$ be a strategy profile that
minimizes $C$.
We define
\begin{equation*}
\PoA(S,G) = \sup\left\{\frac{C(s)}{C(s^*)} : s \in S\right\}
\quad \text{and}\quad
\PoS(S,G) = \inf\left\{\frac{C(s)}{C(s^*)} : s \in S\right\}.
\end{equation*}
The \emph{coarse} (respectively \emph{correlated, mixed, pure}) \emph{price of anarchy} of
a class of games $\mathcal{G}$ is defined as
\begin{equation}\label{poadef}
\sup \{\PoA(S_G,G) : G \in \mathcal{G} \},
\end{equation}
 where $S_G = \cce(G)$ (respectively $\ce(G)$, $\mne(G)$, $\pne(G)$).
The \emph{coarse} (respectively \emph{correlated, mixed, pure}) \emph{price of stability} of a class of games is defined analogously, i.e., by replacing $\PoA$ by $\PoS$ in (\ref{poadef}).
\end{definition}

Notice that the price of anarchy and price of stability are defined with respect to the \emph{original} social cost function $C$, not accounting for the altruistic components.
This reflects our desire to understand the overall performance of the
system (or strategic game), which is not affected by different
\emph{perceptions} of costs by individuals.
Note, however, that if all players have a uniform altruism level
$\alpha_i = \alpha \in [0,1]$ and the social cost function $C$ is equal
to the sum of all players' individual costs, then for every strategy
profile $s \in \Sigma$,
$C^\alpha(s) = (1-\alpha+\alpha n) C(s)$, where $C^\alpha(s) = \sum_{i \in
  N} C_i^\alpha(s)$ denotes the sum of all players' perceived costs. In
particular, bounding the price of anarchy with respect to $C$ is
equivalent to bounding the price of anarchy with respect to total
perceived cost $C^\alpha$ in this case.

We extend Definition~\ref{def:poa} in the obvious way to payoff-maximization games $G$ with social welfare function $\Pi$ by considering the ratio $\Pi(s^*)/\Pi(s)$, where $s^*$ refers to a strategy profile maximizing $\Pi$.

\subsection{Smoothness} \label{sec:smoothness}
Many proofs bounding the price of anarchy for specific games (e.g.,
\cite{roughgarden4,vetta}) use the fact that deviating from an
equilibrium to the strategy at optimum is not beneficial for any
player. The addition of these inequalities, combined with suitable
properties of the social cost function, then gives a bound on the
equilibrium's cost. Roughgarden \cite{roughgarden} recently captured
the essence of this type of argument with his definition of
\emph{$(\lambda, \mu)$-smoothness} of a game, thus providing a generic
template for proving bounds on the price of anarchy.
Indeed, because such arguments only reason about local moves by
players, they immediately imply bounds not only for Nash equilibria,
but all classes of equilibria defined in Section~\ref{sec:equilibrium},
as well as the outcomes of no-regret sequences of play \cite{blum,blum:even-dar}.
Recent work has explored both the limits of this concept
\cite{roughgarden2} and a refinement requiring smoothness only in
local neighborhoods \cite{roughgarden3}. The latter permits more
fine-grained analysis of games, but applies only to correlated
equilibria and their subclasses.

In extending the definition of smoothness to altruistic games,
we have to exercise some care. Simply applying Roughgarden's
definition to the new game does not work, as the social cost function
we wish to bound is the sum of all direct costs \emph{without
  consideration of the altruistic component}. Thus, with respect to
the social cost, altruistic games are in general not sum-bounded.
For this reason, we propose a slightly revised definition of
\emph{$(\lambda,\mu,\altvec)$-smoothness}; the bulk of this paper
is devoted to showing that with the definition, many useful properties
of Roughgarden's definition are preserved.

%Note that games with altruistic players are not cost minimization
%games in general, and therefore it is not possible to bound the price
%of anarchy using the classical definition of
%$(\lambda,\mu)$-smoothness. It turns out that despite this, we are
%still able to extend the concept of $(\lambda,\mu)$-smoothness to
%altruistic extensions of strategic games in a useful way. The
%definition of \emph{$(\lambda,\mu,\altvec)$-smoothness} that we give
%below, is chosen such that it generalizes $(\lambda,\mu)$-smoothness
%and preserves many of its useful properties, as we will see later.

For notational convenience, we define $C_{-i}(s) = C(s) - C_i(s) \le
\sum_{j \neq i} C_j(s)$. Note that when the social cost is the sum of
all players' costs, the inequality is an equality.

\begin{definition}[$(\lambda, \mu, \altvec)$-smoothness]\label{def:smoothness}
Let $G^\altvec$ be a $\altvec$-altruistic extension of a game with
sum-bounded social cost function $C$.
$G^\altvec$ is \emph{$(\lambda, \mu, \altvec)$-smooth} iff
for any two strategy profiles $s, s^* \in \Sigma$,
\begin{equation}
\sum_{i = 1}^n C_i(s_i^*, s_{-i}) + \alpha_i(C_{-i}(s_i^*, s_{-i}) - C_{-i}(s)) \leq \lambda C(s^*) + \mu C(s) \enspace . \label{eq:smoothness}
\end{equation}
\end{definition}
For $\altvec = \mathbf{0}$, this definition coincides with
Roughgarden's notion of $(\lambda, \mu)$-smoothness.
%The definition of $(\lambda, \mu, \altvec)$-smoothness
%may look puzzling at first sight, but the following
%theorem explains it:
To gain some intuition, consider two strategy profiles
$s, s^* \in \Sigma$, and a player $i \in N$ who switches from his
strategy $s_i$ under $s$ to $s^*_{i}$, while the strategies of the
other players remain fixed at $s_{-i}$. The contribution of player $i$
to the left-hand side of \eqref{eq:smoothness} then accounts for the
individual cost that player $i$ perceives after the switch plus
$\alpha_i$ times the difference in social cost
caused by this switch exluding player $i$. The sum of these contributions needs to be
bounded by $\lambda C(s^*) + \mu C(s)$.
We will see that this definition of $(\lambda,\mu,\alpha)$-smoothness allows us to quantify the price of anarchy of some large classes of altruistic games with respect to the very broad class of coarse correlated equilibria.
%\gsdelete{
% GS: Removed this sentence because we argued before why Roughgarden's result cannot be used as a black-box. Also the reason is not that we do not have cost-minimization games (which we have indeed).
%Because these are not cost-minimization games, these results are impossible to obtain by using the classical $(\lambda,\mu)$-smoothness technique.}

\subsection{Preliminary Results}\label{sec:pre-res}

We first show that many of the results in \cite{roughgarden} following
from $(\lambda,\mu)$-smoothness carry over to our altruistic setting
using the extended $(\lambda,\mu,\altvec)$-smoothness notion
(Definition~\ref{def:smoothness}). Even though some care has to be
taken in extending these results, most of the proofs of the
propositions in this section follow along similar lines as their analogues in
\cite{roughgarden}.
%First, the following proposition shows how we can use $(\lambda,\mu,\altvec)$-smoothness for bounding the coarse price of anarchy in $\altvec$-altruistic games.

\begin{proposition}\label{thm:coarsesmoothpoa}
Let $G^\altvec$ be a $\altvec$-altruistic game.
If $G^\altvec$ is $(\lambda, \mu, \altvec)$-smooth with $\mu < 1$,
then the coarse (and thus correlated, mixed, and pure) price of anarchy of
$G^\altvec$ is at most $\frac{\lambda}{1 - \mu}$.
\end{proposition}
\begin{proof}
Let $\sigma$ be a coarse equilibrium of $G^\alpha$, $s$ a random
variable with distribution $\sigma$, and $s^* \in \Sigma$ an
arbitrary strategy profile. The coarse equilibrium condition implies
that for every player $i \in N$:
\begin{equation*}
\Expect{(1 - \alpha_i) C_i(s) + \alpha_i C(s)} \leq \Expect{(1 - \alpha_i) C_i(s_i^*, s_{-i}) + \alpha_i C(s_i^*, s_{-i})} \enspace.
\end{equation*}
By linearity of expectation, for every player $i \in N$:
\begin{equation*}
\Expect{C_i(s)} \leq \Expect{C_i(s_i^*, s_{-i}) + \alpha_i (C(s_i^*, s_{-i}) - C_i(s^*_i, s_{-i})) - \alpha_i(C(s) - C_i(s))} \enspace .
\end{equation*}
By summing over all players and using linearity of expectation, we obtain
\begin{equation*}
\Expect{C(s)} \leq \Expect{\sum_{i=1}^{n} C_i(s_i^*, s_{-i}) + \alpha_i (C_{-i}(s_i^*, s_{-i}) - C_{-i}(s))} \enspace .
\end{equation*}
Now we use the smoothness property (\ref{eq:smoothness}) to conclude
\begin{equation*}
\Expect{C(s)} \leq \Expect{\lambda C(s^*) + \mu C(s)} = \lambda C(s^*) + \mu \Expect{C(s)} \enspace .
\end{equation*}
Solving for $\Expect{C(s)}$ now proves the claim.
As coarse equilibria include correlated equilibria, mixed Nash
equilibria and pure Nash equilibria,
the correlated, mixed, and pure price of anarchy are thus also bounded
by $\frac{\lambda}{1 - \mu}$.
\end{proof}

As we show later, for many important classes of games,
the bounds obtained by $(\lambda, \mu, \altvec)$-smoothness arguments
are actually tight, even for pure Nash equilibria.
Therefore, as in \cite{roughgarden}, we define the
\emph{robust price of anarchy} as the best
possible bound on the coarse price of anarchy obtainable by a
$(\lambda, \mu, \altvec)$-smoothness argument.

\begin{definition}\label{def:rpoa}
%Let $A_G(\alpha)$ be the set of all tuples $(\lambda, \mu, \altvec)$ for which the $\alpha$-altruistic extension $G^\altvec$ of $G$ is $(\lambda, \mu, \altvec)$-smooth with $\mu < 1$.
The \emph{robust price of anarchy} of a $\alpha$-altruistic game $G^\altvec$ is defined as
\begin{equation*}
\RPoA_G(\altvec) = \inf \; \Set{\textstyle{\frac{\lambda}{1-\mu}}}{\text{$G^\alpha$ is $(\lambda, \mu, \altvec)$-smooth, $\mu < 1$}} \enspace .
\end{equation*}
For a class $\mathcal{G}$ of games, we define
$\RPoA_{\mathcal{G}}(\altvec) = \sup\Set{\RPoA_G(\altvec)}{G \in \mathcal{G}}$.
We omit the subscript when the game (or class of games) is clear from the context.
\end{definition}

The smoothness condition also proves useful in the context of no-regret
sequences and the \emph{price of total anarchy}, introduced by Blum et al.~\cite{blum}.

%\gscomment{Should we extend this proof to probability distributions as in Tim's paper?}

\begin{proposition}\label{prop:no-regret}
Let $s^*$ be a strategy profile minimizing the social cost function $C$
of an $\altvec$-altruistic game $G^\alpha$,
and $s^1, \dots, s^T$ a sequence of strategy profiles
in which every player $i \in N$ experiences vanishing average external regret, i.e.,
\begin{equation*}
\sum_{t=1}^T C^\alpha_i(s^t) \le \left(\min_{s'_i \in \Sigma_i} \sum_{t=1}^T C^\alpha_i(s'_i, s^t_{-i})\right) + o(T) \enspace.
\end{equation*}
The average cost of this sequence of $T$ strategy profiles then satisfies
\begin{equation*}
\frac{1}{T} \sum_{t=1}^T C(s^t) \le \RPoA(\alpha) \cdot C(s^*)
\quad \text{as } T \rightarrow \infty.
\end{equation*}
\end{proposition}
\begin{proof}
Consider a sequence $s^1, \dots, s^T$ of strategy profiles of an
$\altvec$-altruistic game $G^{\altvec}$ that is $(\lambda, \mu,\altvec)$-smooth
with $\mu < 1$.
For every $i \in N$ and $t \in \{1, \dots, T\}$, define
\begin{equation*}
\delta^{\altvec}_i(s^t) = C^\alpha_i(s^t) - C^\alpha_i(s^*_i, s^t_{-i}) \enspace .
\end{equation*}
Let $\Delta(s^t) = \sum_{i = 1}^n \delta_i^{\altvec}(s^t)$.
We have
\begin{eqnarray*}
\Delta(s^t) & = & \sum_{i = 1}^n C^\alpha_i(s^t) - C^\alpha_i(s_i^*, s^t_{-i}) \\
& = & \sum_{i=1}^n \left( (1-\alpha_i) C_i(s^t) + \alpha_i C(s^t) - \left( (1-\alpha_i) C_i(s^*_i, s^t_{-i}) + \alpha_i C(s^*_i, s^t_{-i})\right) \right) \\
& = & C(s^t) - \sum_{i=1}^n \left( C_i(s^*_i, s^t_{-i}) + \alpha_i ( C_{-i}(s^*_i, s^t_{-i}) - C_{-i}(s^t) )\right) \enspace .
\end{eqnarray*}

Exploiting the $(\lambda, \mu,\altvec)$-smoothness property, we obtain
\begin{eqnarray}\label{eq:st-delta}
C(s^t) & \le & \textstyle\frac{\lambda}{1-\mu} C(s^*) + \frac{1}{1-\mu} \Delta(s^t) \enspace .
\end{eqnarray}
Suppose that $s^1, \dots, s^T$ is a sequence of strategy profiles in which every player experiences vanishing average external regret, i.e.,
\begin{equation*}
\sum_{t=1}^T C^\alpha_i(s^t) \le \left(\min_{s'_i \in \Sigma_i} \sum_{t=1}^T C^\alpha_i(s'_i, s^t_{-i})\right) + o(T) \enspace .
\end{equation*}
We obtain that for every player $i \in N$:
\begin{equation*}
\frac{1}{T} \sum_{t=1}^T \delta^\alpha_i(s^t) \le \frac{1}{T} \left(\sum_{t=1}^T C^\alpha_i(s^t) - \min_{s'_i \in \Sigma_i} \sum_{t = 1}^T C^\alpha_i(s'_i, s^t_{-i})\right) = o(1) \enspace .
\end{equation*}
Using this inequality and \eqref{eq:st-delta}, we obtain that the average cost of the sequence of $T$ strategy profiles is
\begin{equation*}
\frac{1}{T} \sum_{t=1}^T C(s^t) \leq \frac{\lambda}{1-\mu} C(s^*) + \frac{1}{1-\mu} \sum_{i = 1}^n \left( \frac{1}{T} \sum_{t=1}^T \delta^\alpha_i(s^t) \right) \overset{T \rightarrow \infty}{\longrightarrow}\frac{\lambda}{1-\mu} C(s^*) \enspace .
\end{equation*}
\end{proof}

Roughgarden \cite[Proposition 2.6]{roughgarden} shows that for games that
have an \textit{underestimating exact potential function}, best
response dynamics\footnote{Best response dynamics are a natural way of
  searching for a pure Nash equilibrium: if the current strategy
  profile is not a Nash equilibrium, then pick a player who can
  improve his cost and change his strategy to one that minimizes his
  cost.} converge rapidly to a strategy profile of social cost close
to the robust price of anarchy times  the optimum social cost of the
game; see \cite{roughgarden} for a precise statement of this result
and the accompanying definitions. Proposition 2.6 in
\cite{roughgarden} and its proof straightforwardly carry over to
$(\lambda,\mu,\altvec)$-smooth games that have such an underestimating
exact potential function.
%\gscomment{At the latest in the journal version we should provide a formal proof for this}

The results in this section continue to hold for altruistic extensions
of payoff-maximization games if we adapt Definition
\ref{def:smoothness} as follows.
Let $G^\alpha$ be an $\alpha$-altruistic extension of a
payoff-maximization game with social welfare function $\Pi$.
Define $\Pi_{-i}(s) = \Pi(s) - \Pi_i(s)$. $G^\alpha$ is $(\lambda,\mu,\altvec)$-smooth iff for every two strategy profiles $s, s^* \in \Sigma$,
\begin{equation}
\sum_{i = 1}^n (\Pi_i(s_i^*, s_{-i}) + \alpha_i(\Pi_{-i}(s_i^*, s_{-i}) - \Pi_{-i}(s))) \geq \lambda \Pi(s^*) - \mu \Pi(s) \enspace. \label{eq:smooth-max}
\end{equation}
Given this smoothness definition, all the results above hold when we
replace $\frac{\lambda}{1 - \mu}$ by $\frac{1 + \mu}{\lambda}$ and
$\mu < 1$ by $\mu > -1$ in Definition \ref{def:rpoa}.

\section{Fair Cost-sharing Games}\label{sec:csg}

A \emph{fair cost-sharing game}
$G = (N, E, \{\Sigma_i\}_{i \in N}, \{c_e\}_{e \in E})$
is characterized by a set $E$ of \emph{resources} (or \emph{facilities}),
and strategy sets $\Sigma_i \subseteq 2^E$; that is,
players' strategies $s_i \in \Sigma_i$ are subsets of resources.
Given a strategy profile
$s \in \Sigma = \Sigma_1 \times \dots \times \Sigma_n$,
we define $x_e(s) = |\{ i \in N \; : \; e \in s_i\}|$ as the number of
players that use resource $e \in E$ under $s$.
Let $U(s)$ be the set of resources that are used under $s$, i.e., $U(s) = \bigcup_{i \in N} s_i$.
%\gscomment{replaced $E(s)$ by $U(s)$ in this section to be consistent with the section on valid utility games}
Each facility $e \in E$ has a cost $c_e$ which is evenly shared by all
players using $e$, i.e., the direct cost of player $i$ is $C_i(s) = \sum_{e \in s_i} c_e/x_e(s)$.
The social cost function is $C(s) = \sum_{i=1}^{n} C_i(s) = \sum_{e \in U(s)} c_e$.

It is well-known that the pure price of anarchy of fair cost-sharing games is $n$ \cite{nisan}.
We show that it can get significantly worse in the presence of
altruistic players: the following theorem gives a much worse
upper bound, which we subsequently show to be tight.
%Here we study $\altvec$-altruistic extensions of these games.
\begin{theorem}\label{thm:cspoauniform}
The robust price of anarchy of $\altvec$-altruistic cost-sharing games
is at most $\frac{n}{1-\maxbeta}$ (with $n/0 = \infty$).
\end{theorem}

\begin{proof}
The claim is true for $\maxbeta = 1$ because $\RPoA(\altvec) \le \infty$ holds trivially.
We show that $G^\altvec$ is $(n, \maxbeta, \altvec)$-smooth for $\maxbeta \in [0,1)$.
Let $s$ and $s^*$ be two strategy profiles.
Fix an arbitrary player $i \in N$. We have
\begin{equation*}
C(s_i^*, s_{-i}) - C(s) = \sum_{e \in U(s_i^*, s_{-i})} c_e - \sum_{e \in U(s)} c_e \le \sum_{e \in s^*_i \setminus U(s)} c_e \enspace .
\end{equation*}
We use this inequality to obtain the following bound:
\begin{align*}
(1 - \alpha_i)C_i(s_i^*, s_{-i}) + \alpha_i( C(s_i^*, s_{-i}) - C(s))
& \leq (1-\alpha_i) \sum_{e \in s^*_i} \frac{c_e}{x_e(s_i^*, s_{-i})} + \alpha_i \sum_{e \in s^*_i \setminus U(s)} \frac{c_e}{x_e(s_i^*, s_{-i})} \\
& \leq \sum_{e \in s^*_i} \frac{c_e}{x_e(s_i^*, s_{-i})} \leq \sum_{e \in s^*_i} \frac{n \cdot c_e}{x_e(s^*)} \enspace.
\end{align*}
The first inequality holds because $x_e(s^*_i, s_{-i}) = 1$ for every $e \in s^*_i \setminus U(s)$, and
the last inequality follows from $x_e(s^*_i, s_{-i}) \ge x_e(s^*)/n$ for every $e \in s^*_i$.
The left-hand side of the smoothness condition \eqref{eq:smoothness} is equivalent to
\begin{align*}
\sum_{i = 1}^n \left((1 - \alpha_i)C_i(s_i^*, s_{-i}) + \alpha_i( C(s_i^*, s_{-i}) - C(s)) + \alpha_i C_i(s) \right)
& \leq \sum_{i = 1}^n \left(\sum_{e \in s^*_i} \frac{n \cdot c_e}{x_e(s^*)} \right) + \maxbeta C(s) \\
& = n C(s^*) + \maxbeta C(s) \enspace .
\end{align*}
We conclude that the robust price of anarchy is at most
$\frac{n}{1-\maxbeta}$.
Example~\ref{lem:costsharinglb} shows that this bound is tight, even
for pure Nash equilibria.
% \qed
\end{proof}

\begin{example}\label{lem:costsharinglb}
Consider the cost-sharing game in which $n$ players can choose between
two different facilities $e_1$ and $e_2$ of cost $1$ and
$n/(1-\alpha)$, respectively. Let $s^* = (e_1, \dots, e_1)$ and
$s = (e_2, \dots, e_2)$ refer to the strategy profiles in which every
player chooses $e_1$ and $e_2$, respectively. Then $C(s^*) = 1$ and
$C(s) = n/(1-\alpha)$. Note that $s$ is a pure Nash equilibrium of the
$\alpha$-altruistic extension of this game because for every player
$i$ we have
\begin{align*}
%C_i^\alpha(s) & =
(1-\alpha) C_i(s) + \alpha C(s) = 1+ \alpha \frac{n}{1-\alpha} = C_i^\alpha(\{e_1\}, s_{-i}) \enspace.
\end{align*}
The pure price of anarchy is therefore at least $n/(1-\alpha)$.
\end{example}

We turn to the pure price of stability of uniformly $\alpha$-altruistic cost-sharing games.
Clearly, an upper bound on the pure price of stability extends to the mixed, correlated and coarse price of stability.
As opposed to the price of anarchy, the price of stability does improve with increased altruism.
The proof of the following proposition exploits a standard technique to bound the pure price of stability of exact potential games (see, e.g., \cite{nisan}).

\begin{proposition}\label{pro:poscsg}
The pure price of stability of uniformly $\alpha$-altruistic cost-sharing games is at most $(1-\alpha)H_n + \alpha$.
\end{proposition}
\begin{proof}
Let $G^\alpha$ be a uniformly $\alpha$-altruistic cost-sharing game. It is not hard to verify that $G^\alpha$ is an exact potential game with potential function $\Phi^\alpha(s) = (1-\alpha) \Phi(s) + \alpha C(s)$, where $\Phi(s) = \sum_{e \in E} \sum_{i = 1}^{x_e(s)} c_e /i$.
Observe that
\begin{align*}
\Phi^\alpha(s)=(1-\alpha)\sum_{e \in E} \sum_{i = 1}^{x_e(s)} \frac{c_e}{i} +\alpha\sum_{e \in U(s)}c_e\le((1-\alpha)H_n+\alpha)\sum_{e \in U(s)} c_e=((1-\alpha)H_n+\alpha)C(s).
\end{align*}
We therefore have that $C(s) \le \Phi^\alpha(s) \le ((1-\alpha)H_n + \alpha) C(s)$.

Let $s$ be a strategy profile that minimizes $\Phi^\alpha$, and
let $s^*$ be an optimal strategy profile that minimizes the social cost function $C$.
Note that $s$ is a pure Nash equilibrium of $G^\alpha$.
We have
\begin{equation*}
C(s) \le \Phi^\alpha(s) \le \Phi^\alpha(s^*) \le
((1-\alpha)H_n + \alpha) C(s^*) \enspace,
\end{equation*}
which proves the claim.
\end{proof}

\section{Valid Utility Games}\label{sec:vug}

A valid utility game \cite{vetta} is a payoff maximization game given
by $G = (N, E, \{\Sigma_i\}_{i \in N}, \{\Pi_i\}_{i \in N}, V)$, where
$E$ is a ground set of resources, the strategy sets $\Sigma_i$
are subsets of $E$, $\Pi_i$ is the payoff function of player $i$, and $V$
is a submodular\footnote{For a finite
  set $E$, a function $f : 2^E \rightarrow \mathbb{R}$ is
  \emph{submodular} iff $f(A \cup \{x\}) - f(A) \geq f(B \cup \{x\}) -
  f(B)$ for any $A \subseteq B \subseteq E, x \in E$.} and
non-negative function on $E$. Every player strives to maximize his
individual payoff function $\Pi_i$.

For a strategy profile $s \in \Sigma$, let
$U(s) = \bigcup_{i \in N} s_i \subseteq E$ be
the union of all players' strategies under $s$. The social welfare
function $\Pi: \Sigma \rightarrow \mathbb{R}$ to be maximized
is $\Pi(s) = V(U(s))$, and thus depends only on the union of the
players' chosen strategies, evaluated by $V$.
The individual payoff functions of all
players $i \in N$ are assumed to satisfy\footnote{We abuse notation and
  write $\Pi(\varnothing, s_{-i})$ to denote $V(U(s) \backslash
  s_i)$.} $\Pi_i(s) \ge \Pi(s) - \Pi(\varnothing, s_{-i})$ for every
strategy profile $s \in \Sigma$. Intuitively, this means that the
individual payoff of a player is at least his contribution to the
social welfare. Moreover, it is assumed that $\Pi(s) \ge \sum_{i = 1}^n
\Pi_i(s)$ for every $s \in \Sigma$. See \cite{vetta} for a detailed
description and justification of these assumptions.

Examples of games falling into this framework include natural
game-theoretic variants of the facility location,
$k$-median and network routing problems \cite{vetta}.
Vetta \cite{vetta} proved a bound of $2$ on the pure price of anarchy
for valid utility games with non-decreasing $V$, and Roughgarden
showed in \cite{roughgarden} how this bound is achieved via a
$(\lambda, \mu)$-smoothness argument. We extend this result to
altruistic extensions of these games.

%\pcedit{Different from congestion games and cost sharing games, the robust price of anarchy of $\altvec$-altruistic valid utility games is $2$, regardless of the altruism levels of the players. There is related interesting results about \emph{Byzantine} players, about whose behavior no assumptions can be made. Blum et al. \cite{blum} have showed that the coarse price of anarchy is 2 in valid games even when arbitrarily many Byzantine players show up in the system.}

\begin{theorem}\label{thm:vugpoa}
The robust price of anarchy of $\altvec$-altruistic valid utility games is $2$.
\end{theorem}

Note that in the statement above, $\altvec$ is a vector, i.e., the claim holds for arbitrary non-uniform altruism.

\begin{proof}
We show that the $\altvec$-altruistic extension $G^{\altvec}$ of a
valid utility game is $(1, 1, \altvec$)-smooth.

Fix two strategy profiles $s, s^* \in \Sigma$ and consider an
arbitrary player $i \in N$. By assumption, we have
$\Pi_i(s) \geq \Pi(s) - \Pi(\varnothing, s_{-i})$.
Therefore, for each player $i \in N$,
\begin{eqnarray}
\Pi(s_i^*, s_{-i}) - \Pi(s) + \Pi_i(s) & = & (\Pi(s_i^*, s_{-i}) - \Pi(\varnothing, s_{-i})) \notag
- (\Pi(s) - \Pi(\varnothing, s_{-i})) + \Pi_i(s) \notag \\
%& \quad \geq \Pi(s_i^*, s_{-i}) - \Pi(\varnothing, s_{-i})) - \alpha_i\Pi_i(s) + \alpha_i\Pi_i(s) \notag \\
& \geq & \Pi(s_i^*, s_{-i}) - \Pi(\varnothing, s_{-i}) \enspace . \label{eq:qrs}
\end{eqnarray}
Now let $U_i = \bigcup_{j=1}^n s_j \cup \bigcup_{j=1}^i s_j^*$. Summing over all $i \in N$,
\begin{eqnarray*}
\sum_{i=1}^n ((1-\alpha_i) \Pi_i(s_i^*, s_{-i}) + \alpha_i(\Pi(s_i^*, s_{-i}) - \Pi(s) + \Pi_i(s)))
& \geq & \sum_{i=1}^n (\Pi(s_i^*, s_{-i}) - \Pi(\varnothing, s_{-i})) \\
& = & \sum_{i=1}^n (V(U(s_i^*, s_{-i})) - V(U(s) \setminus s_i)) \\
& \geq & \sum_{i=1}^n (V(U_i) - V(U_{i-1})) \\
& \geq &\Pi(s^*) - \Pi(s) \enspace .
\end{eqnarray*}
Here, the first inequality follows from \eqref{eq:qrs} and because
$\Pi_i(s) \ge \Pi(s) - \Pi(\varnothing, s_{-i})$ for every $i$,
the second inequality holds because $V$ is submodular, and the
final inequality follows from $V$ being non-decreasing.
We conclude that $G^\altvec$ is $(1, 1, \altvec)$-smooth,
which proves an upper bound of $2$ on the robust price of anarchy.
This bound is tight, as shown by Example~\ref{lem:validutilitytight}.
%\qed
\end{proof}

\begin{example}\label{lem:validutilitytight}
Consider a valid utility game $G$ with two players $N = \{1,2\}$,
a ground set $E = \{1,2\}$ of two elements and strategy sets
$\Sigma_1 = \{\{1\}, \{2\}\},\ \Sigma_2 = \{ \varnothing, \{1\} \}$.
Define $V(S) = |S|$ for every subset $S \subseteq E$. Note that $V$ is non-negative, non-decreasing and submodular.

For a given strategy profile $s \in \Sigma$, the individual profits
$\Pi_1(s)$ and $\Pi_2(s)$ of player 1 and player 2, respectively, are
defined as  follows:
$\Pi_1(s) = 1$ for all strategy profiles $s$.
$\Pi_2(s) = 1$ if $s = (\{2\},\{1\})$ and $\Pi_2(s) = 0$ otherwise.
It is not hard to verify that for every player $i$ and every strategy
profile $s \in \Sigma$ we have $\Pi_i(s) \geq \Pi(s) -
\Pi(\varnothing, s_{-i})$. Moreover, $\Pi(s) \geq \Pi_1(s) + \Pi_2(s)$
for every $s \in \Sigma$.
We conclude that $G$ is a valid utility game.

Let $\altvec \in [0,1]^2$, and consider the $\altvec$-altruistic
extension $G^\altvec$ of $G$. We claim that $s = (\{1\}, \varnothing)$
is a pure Nash equilibrium of $G^\altvec$: the profit of player $1$
under $s$ is $(1 - \alpha_1) + \alpha_1 = 1$. His profit remains $1$ if
he switches to strategy $\{2\}$. The profit of player $2$ under $s$ is
$\alpha_2$. If he switches to strategy $\{1\}$, then his profit is
$\alpha_2$ as well. Thus, $s$ is a pure Nash equilibrium. Since $\Pi(s)
= 1$ and $\Pi((\{2\}, \{1\})) = 2$, the pure price of anarchy of $G$ is $2$.
\end{example}

\section{Congestion Games}\label{sec:cg}
In an \emph{atomic congestion game}
$G = (N, E, \{\Sigma_i\}_{i \in N}, \{d_e\}_{e \in E})$,
players' strategies are again subsets of facilities,
$\Sigma_i \subseteq 2^E$.
Each facility $e \in E$ has an associated
\emph{delay function} $d_e : \mathbb{N} \rightarrow \mathbb{R}$.
As in Section \ref{sec:csg}, we write $x_e(s)$ for the number of
players using facility $e$.
Player $i$'s cost is $C_i(s) = \sum_{e \in s_i} d_e(x_e(s))$,
and the social cost is $C(s) = \sum_{i=1}^{n} C_i(s)$.
We focus on \emph{linear} congestion games, i.e., the delay functions are of the
form $d_e(x) = a_ex + b_e$, where $a_e, b_e$ are non-negative rational numbers.
Pure Nash equilibria of altruistic extensions of linear congestion games
always exist \cite{skopalik}; this may
not be the case for arbitrary (non-linear) congestion
games.

The PoA of linear congestion games is known to be $\frac{5}{2}$
\cite{christodoulou}.
Recently, Caragiannis et al.~\cite{caragiannis} extended this result
to linear congestion games with uniformly altruistic players. Applying
the transformation outlined in Remark~\ref{rem:rel-Car}, their
result can be stated as follows:

\begin{theorem}[Caragiannis et al. \cite{caragiannis}]\label{thm:cara}
The pure price of anarchy of uniformly $\alpha$-altruistic linear congestion games
is at most $\frac{5+4\alpha}{2+\alpha}$.
\end{theorem}

The proof in \cite{caragiannis} implicitly uses a smoothness argument
in the framework we define here for altruistic games.
Thus, without any additional work, our framework allows the extension
of Theorem \ref{thm:cara} to the robust PoA.
Caragiannis et al.~\cite{caragiannis} also showed that the bound of Theorem \ref{thm:cara} is asymptotically tight.
A simpler example (given below) proves tightness of this bound (not only asymptotically). Thus, the robust
price of anarchy is exactly $\frac{5+4\alpha}{2+\alpha}$.
We give a refinement of Theorem \ref{thm:cara} to non-uniform altruism distributions,
obtaining a bound in terms of the maximum and minimum altruism levels.

\begin{theorem}\label{thm:cgpoa}
The robust price of anarchy of $\altvec$-altruistic linear congestion
games is at most
$\frac{5+2\maxbeta+2\minbeta}{2-\maxbeta+2\minbeta}$.
\end{theorem}

%If all players have uniform altruism level $\alpha \in [0,1]$, the
%above bound reduces to $\frac{5+4\alpha}{2+\alpha}$.
%This is tight, as shown in Example~\ref{lem:congestiongamelb} below.
%\begin{remark}
%The pure price of anarchy with uniformly altruism of
%linear congestion games in \cite{caragiannis} is $\frac{5-\xi}{2-\xi}$ for $\alpha(1-\xi)=\xi$.
%\end{remark}

As a first step, we show that without loss of generality, we can focus
on simpler instances of linear congestion games.

\begin{lemma}
Without loss of generality, all delay functions are of the form
$d_e(x) = x$.
\end{lemma}

\begin{proof}
First, we may assume that for every delay function $d_e$, the $a_e$
and $b_e$ coefficients are integers.
This can be ensured by multiplying all coefficients among all
facilities by their least common multiple.
In the resulting game, all coefficients are integers, the price of
anarchy is the same, and so is the set of all equilibria.

Next, we can assume that $b_e = 0$ for all $e \in E$.
To show this, we replace any facility $e \in E$ with delay function
$d(x) = a_e x + b_e$ by $n+1$ facilities
$e_0, \dots, e_n$ with delay functions $d_{e_0}(x) = a_e x$ and
$d_{e_i}(x) = b_e x$ for $1 \leq i \leq n$.
We then adapt the strategy space $\Sigma_i$ of each player $i$ as
follows: we replace every strategy $s_i \in \Sigma_i$ in which $e$
occurs by the strategy $s_i \setminus \{e\} \cup \{e_0, e_i\}$.
There is an obvious bijection between the strategy profiles in the
original game and those in the new game, preserving the values of
individual cost functions and the social cost function.
(Notice that this construction exploits the fact that all players have unit weight,
and would not carry over to weighted congestion games.)

Finally, for the same reason, we can also assume that $a_e = 1$ for all
$e \in E$. We replace $e$ with facilities $e_1, \ldots, e_{a_e}$, each
having delay function $d_{e_i}(x) = x$, and adapt the strategy space
$\Sigma_i$ of each player $i$ by replacing each strategy $s_i$ in
which $e$ occurs by $s_i \setminus \{e\} \cup \{e_1, \ldots, a_{a_e}\}$.
Now, all delay functions are $d_e(x)=x$.
\end{proof}

The next step in the proof of Theorem~\ref{thm:cgpoa} is the following
technical lemma:
\begin{lemma}\label{lem:calc}
For every two non-negative integers $x, y$ and $\maxbeta, \minbeta \in [0,1]$ with $\maxbeta \ge \minbeta$,
\begin{equation*}
((1+ \maxbeta)x + 1)y + \minbeta (1-x) x \le \frac{5+2\maxbeta+2\minbeta}{3} y^2 + \frac{1+\maxbeta-2\minbeta}{3} x^2.
\end{equation*}
\end{lemma}

To prove this lemma, we make use of the following result:
\begin{lemma}\label{lem:calc2}
For all $x,y \in \mathbb{N}_0$, $\alpha \in [0,1]$ and $\beta \in [0,1]$, it holds that
\begin{equation*}
((1+\alpha)x+1)y + \beta\alpha(1-x)x \le (2+\alpha-\gamma) y^2 + \gamma x^2
\end{equation*}
for all $\gamma \in \textstyle[\frac13 (1+\alpha-2\beta\alpha),1+\alpha]$.

\end{lemma}
\begin{proof}
The inequality is equivalent to
\begin{equation*}
((1+\alpha)x+1)y + \beta\alpha(1-x)x -(2+\alpha) y^2 \leq \gamma (x^2-y^2) \enspace .
\end{equation*}
Assume that $x=y$. The inequality is then trivially satisfied because $x \le x^2$ for all $x \in \mathbb{N}_0$.
Next suppose that $x > y$. Then
\begin{equation*}
\gamma \geq \frac{((1+\alpha)x+1)y + \beta\alpha(1-x)x -(2+\alpha) y^2}{x^2-y^2} \enspace .
\end{equation*}
We show that the maximum of the expression on the right-hand side is
attained by $x = 2$ and $y = 1$. First, we fill in these values and
conclude that for these values,
$\gamma \geq \frac13(1+\alpha-2\beta\alpha) \geq 0$.
We now write $x$ as $y + a, a \geq 1$,
and rewrite the right-hand side as
\begin{equation}\label{eq:lem1rhs1}
f(y,a) = \frac{(1 + \alpha)y + \beta \alpha}{2y + a} + \frac{(1 + \beta\alpha)(y - y^2)}{a(2y + a)} - \beta\alpha \enspace .
\end{equation}
Because we know that there are choices of $x$ and $a$ for which $f(y,
a)$ is positive (e.g., when $y = 1$ and $a = 1$), and because $a$ only
occurs in the denominators, we know that \eqref{eq:lem1rhs1} reaches
its maximum when $a = 1$. So we assume $a = 1$. When we then fill in
$y = 0$, we see that $f(0,1) = 0$, so $f(1,1) \geq f(0,1)$. When $y >
1$ we can write $y$ as $w + 2$, where $w \geq 0$, and we can now
further rewrite $f(y,a)$ as
\begin{equation*}
f(w + 2, 1) = \frac{2\alpha - 6\beta\alpha}{2w+5} - \frac{(2 - \alpha + 5\beta\alpha)w + (1 + \beta\alpha)w^2}{2w + 5} \leq \frac{2\alpha - 6\beta\alpha}{2w+5} \enspace .
\end{equation*}
When $2\alpha - 6\beta\alpha$ is negative, this term is certainly less
than $f(1,1)$. When $2\alpha - 6\beta\alpha$ is positive, we have
\begin{equation*}
f(w + 2, 1) \leq \frac{2\alpha - 6\beta\alpha}{2w+5} \leq \frac{2\alpha - 6\beta\alpha}{5} \leq \frac13(2\alpha - 6\beta\alpha) \leq \frac13(1 + \alpha - 2\beta\alpha) = f(1,1) \enspace .
\end{equation*}
This shows that $\gamma \geq f(1,1) = \frac13(1 + \alpha - 2\beta\alpha)$.

The final case is when $x < y$. Then,
\begin{equation*}
\gamma \leq \frac{(2+\alpha) y^2 - ((1+\alpha)x+1)y - \beta\alpha(1-x)x}{y^2-x^2} \enspace .
\end{equation*}
We show that the minimum of the expression on the right-hand side is
attained by $x = 0$ and $y = 1$. First, we fill in these values and
conclude that for these values, $\gamma \leq 1+\alpha$.
We now write $y$ as $x + a$, $a \geq 1$, and rewrite the
right-hand side as
\begin{equation*}
g(x,a) = \frac{(1 + \beta \alpha)x^2 - (1 + a + (a + \beta)\alpha)x - a}{a(2x + a)} + 2 + \alpha \enspace .
\end{equation*}
Suppose first that $x = 0$ and that $a \geq 2$.
Then we can write $a$ as $1 + b$, $b > 0$, and therefore
\[
f(0, 1 + b) = 2 + \alpha - \frac{1}{1 + b} \geq \frac{3}{2} + \alpha
\geq 1 + \alpha = f(0,1).
\]
When $x \geq 1$, we can write $x$ as $1 + b$, $b \geq 0$. We then have
\begin{equation*}
f(1 + b, a) = 2 + \alpha - \frac{2 + \alpha + (1 - \alpha) b}{2b + 2 + a} + \frac{(1 + \beta\alpha)(b^2 + b)}{a(2b + 2 + a)} \enspace .
\end{equation*}
The last of these terms is positive, hence
\begin{eqnarray*}
f(1 + b, a) & \geq & 2 + \alpha - \frac{2 + \alpha + (1 - \alpha) b}{2b + 2 + a} \geq 2 + \alpha - \frac{2 + 1 + b}{2b + 2 + a} \\
 & \geq & 2 + \alpha - 1 = 1 + \alpha = f(0,1) \enspace.
\end{eqnarray*}
This shows that $\gamma \leq f(0,1) = 1 + \alpha$.
%\qed
\end{proof}

\noindent Now we can complete the proof of Lemma \ref{lem:calc}.

\begin{proof}[Proof of Lemma~\ref{lem:calc}]
Choose $\beta \in [0,1]$ such that $\minbeta = \beta\maxbeta$.
Using Lemma~\ref{lem:calc2} above, we obtain
\begin{equation*}
((1+ \maxbeta)x + 1)y + \minbeta (1-x) x = ((1+ \maxbeta)x + 1)y + \beta\maxbeta (1-x) x \leq (2+\maxbeta-\gamma) y^2 + \gamma x^2 \enspace ,
\end{equation*}
where $\gamma \in [\frac13 (1+\maxbeta-2\beta\maxbeta), 1+\maxbeta]$.
By choosing $\gamma = \frac13(1+\maxbeta-2\beta\maxbeta)$, we obtain
\begin{equation*}
((1+ \maxbeta)x + 1)y + \minbeta (1-x) x \leq \frac{5+2\maxbeta + 2\beta\maxbeta}{3}y^2 + \frac{1+\maxbeta-2\beta\maxbeta}{3}x^2 \enspace.
\end{equation*}
Substituting $\beta\maxbeta = \minbeta$ yields the claim.
%\qed
\end{proof}
We remark that the choice of $\gamma$ in the proof above has been made
in order to minimize the expression $\lambda/(1-\mu)$ (which is an
increasing function in $\gamma$).

Lemma \ref{lem:calc} is essentially the part that generalizes the
proof in \cite{caragiannis}, and allows us to complete the proof
of Theorem \ref{thm:cgpoa}.

\begin{proof}[Proof of Theorem~\ref{thm:cgpoa}]
%Let $G^\alpha$ be a $\alpha$-altruistic linear congestion game.
We show that the $\alpha$-altruistic extension $G^\alpha$ of a linear congestion game is $(\frac13(5+2\maxbeta+2\minbeta), \frac{1}{3} (1+\maxbeta-2\minbeta), \altvec)$-smooth.
%The proof then follows from Proposition~\ref{cor:poa}.

Let $s$ and $s^*$ be two strategy profiles, and write $x_e = x_e(s), x_e^* = x_e(s^*)$.
The left-hand side of the smoothness condition \eqref{eq:smoothness} is equivalent to
%\gscomment{I find that the first equality below is a rather big step. Would be nice to give some intermediate steps to help the reader. I believe to remember that we had written out some more steps in a previous version but could not recover them from previous files..}
\begin{align*}
& \sum_{i=1}^n \left((1 - \alpha_i)C_i(s_i^*, s_{-i}) + \alpha_i( C(s_i^*, s_{-i}) - C(s)) + \alpha_i C_i(s) \right) \\
& = \sum_{i=1}^n \left((1 - \alpha_i)\Bigg(\sum_{e \in s_i^* \backslash s_i} (x_e+1) + \sum_{e \in s_i \cap s_i^*} x_e\Bigg)
 + \alpha_i \left(\sum_{e \in s_i^* \setminus s_i} (2x_e + 1) + \sum_{e \in s_i \backslash s_i^*} (1 - 2x_e)\right) + \alpha_i C_i(s) \right) \\
% & \quad = \sum_{i=1}^n \left(\sum_{e \in s_i^* \backslash s_i} ((1 + \alpha_i)x_e + 1) + (1-\alpha_i)\sum_{s_i \cap s_i^*} x_e \right.\\
% & \qquad  \left .+ \alpha_i\sum_{e \in s_i \backslash s_i^*} (1 - 2x_e) + \alpha_i \sum_{e \in s_i} x_e\right) \\
 & \leq
 %\sum_{i=1}^n \left(\sum_{e \in s_i^*} ((1 + \alpha_i)x_e + 1) + \alpha_i \sum_{e \in s_i} (1-2x_e)+ \alpha_i \sum_{e \in s_i} x_e \right) \\
 %& \quad =
 \sum_{i=1}^n \left(\sum_{e \in s_i^*} ((1 + \alpha_i)x_e + 1) + \alpha_i \sum_{e \in s_i} (1-x_e) \right) \\
 & \leq \sum_{e \in E} \left(((1 + \maxbeta)x_e + 1) x^*_e + \minbeta(1-x_e) x_e \right) \enspace .
\end{align*}
In the above derivation, the first inequality follows from the fact that
$(1 - \alpha_i)x_e \leq (1 + \alpha_i)x_e + 1 + \alpha_i(1 - 2x_e)$ for every $e \in s_i \cap s^*_i$. Therefore, it is possible to simply replace all the $(1 - \alpha_i)x_e$ (in the third summation operator of the left hand side of the first inequality) by $(1 + \alpha_i)x_e + 1 + \alpha_i(1 - 2x_e)$, write $C_i(s)$ as $\sum_{e \in s_i} x_e$, and finally rewrite the resulting expression into the form of the right hand side of the first inequality.
The second inequality holds because for every $i \in N$ and $e \in
s_i$, $1-x_e \le 0$ and by the definition of $\maxbeta$ and
$\minbeta$.
The bound on the robust price of anarchy now follows from Lemma~\ref{lem:calc}.
\end{proof}

%\bkdelete{\paragraph{Tight Example.}}
%\bkcomment{I deleted here the ``tight example'' heading because it is typeset in the same style as ``example 3'', and that looks weird.}
The following is a simple example that shows that the bound of
$\frac{5+4\alpha}{2+\alpha}$ on the robust price of anarchy for uniformly
$\alpha$-altruistic linear congestion games is tight, even for pure
Nash equilibria.
It slightly improves the lower bound example of \cite{caragiannis}, because it is simpler and it shows tightness of the bound not only asymptotically.

\begin{example}\label{lem:congestiongamelb}
Consider a game with
six resources $E = E_1 \cup E_2$, $E_1 = \{h_0, h_1, h_2\}$,
$E_2 = \{g_0, g_1, g_2\}$
and three $\alpha$-altruistic players.
The delay functions are given by $d_e(x)=(1+\alpha)x$ for $e \in E_1$,
and $d_e(x) = x$ for $e \in E_2$.
Each player $i$ has two pure strategies: $\{h_{i-1},g_{i-1}\}$ and
$\{h_{(i-2)\text{ (mod 3)}},h_{i\text{ (mod 3)}},g_{i\text{ (mod 3)}}\}$.
The strategy profile in which every player selects his first strategy
is a social optimum of cost $(1+\alpha) \cdot 3+3=(2+\alpha) \cdot 3$.

Consider the strategy profile $s$ in which every player chooses his
second strategy. We argue that $s$ is a Nash equilibrium.
Each player's perceived individual cost is
$c_1 = (1-\alpha)(4(1+\alpha)+1)+\alpha(5+4\alpha) \cdot 3$, whereas
if a player unilaterally deviates to his first strategy, the new
social cost would become $(5+4\alpha) \cdot 3+1-\alpha$.
Thus, the player's new perceived individual cost is
$c_2 = (1-\alpha)(3(1+\alpha)+2)+\alpha((5+4\alpha) \cdot 3+1-\alpha)$.
Because $c_1=c_2$, $s$ is a Nash equilibrium, of cost
$4(1+\alpha)\cdot 3+3=(5+4\alpha) \cdot 3$.
We conclude that the price of anarchy is at least $\frac{5+4\alpha}{2+\alpha}$ for $\alpha \in [0,1]$.
\end{example}

We turn to the pure price of stability of $\alpha$-altruistic congestion games.
Again, an upper bound on the pure price of stability extends to the
mixed, correlated and coarse price of stability.
%\bkdelete{The proof of the following proposition exploits a standard technique to bound the pure price of stability of exact potential games (see, e.g., \cite{nisan}).}

\begin{proposition}\label{prop:posacg}
The pure price of stability of uniformly $\alpha$-altruistic linear
congestion games is at most $\frac{2}{1+\alpha}$.
\end{proposition}
\begin{proof}
Let $G^\alpha$ be a uniformly $\alpha$-altruistic extension of a linear congestion game. It is not hard to verify that $G^\alpha$ is an exact potential game with potential function $\Phi^\alpha(s) = (1-\alpha) \Phi(s) + \alpha C(s)$, where $\Phi(s) = \sum_{e \in E} \sum_{i = 1}^{x_e(s)} i$ is Rosenthal's potential function.
Observe that
\begin{eqnarray*}
\Phi^\alpha(s) & = & (1-\alpha) \sum_{e \in E} \sum_{i=1}^{x_e(s)} i + \alpha C(s) = \frac{1-\alpha}{2} \sum_{e \in E} (x^2_e(s) + x_e(s)) + \alpha \sum_{e \in E} x_e^2(s) \\
& = & \frac{1+\alpha}{2} C(s) + \frac{1-\alpha}{2} \sum_{e \in E} x_e(s) \enspace .
%\le \frac{1+\alpha}{2} C(s) + \frac{1-\alpha}{2} \sum_{e \in E} x^2_e(s) \le C(s)
\end{eqnarray*}
We therefore have $\frac{1+\alpha}{2} C(s) \le \Phi^\alpha(s) \le C(s)$.
The claim now follows by using similar arguments as in %\gsdelete{\cite{anshelevich} and} the proof of Proposition \ref{pro:poscsg},
which proves the claim.
\end{proof}

\section{Symmetric Singleton Congestion Games}

\emph{Symmetric singleton congestion games}
are an important special case of congestion games.
They are defined as $G = (N, E, \{\Sigma_i
\}_{i \in N}, \{ d_e\}_{e \in E})$: every player chooses one
facility (also called \emph{edge}) from $E = \{1,...,m\}$, and all strategy
sets are identical, i.e., $\Sigma_i = E$ for every $i$.
We refer to these games simply as \emph{singleton congestion games} below.
In \emph{singleton linear congestion games}, the focus here, delay
functions are also assumed to be linear, of the form $d_e(x) = a_e x + b_e$.

\subsection{\protect Uniform Altruism}

Caragiannis et al.~\cite{caragiannis} prove the following theorem
(stated using the transformation from Remark~\ref{rem:rel-Car}). It
shows that the pure price of anarchy does not always increase with the
altruism level; the relationship between $\alpha$ and the price of anarchy is thus
rather subtle.

\begin{theorem}[Caragiannis et al.~\cite{caragiannis}] \label{thm:poa-sscg}
The pure price of anarchy of uniformly $\alpha$-altruistic singleton linear congestion games is $\frac{4}{3+\alpha}$.
\end{theorem}

We show that even the mixed price of anarchy (and thus also the robust price of anarchy) will be at least
2 regardless of the altruism levels of the players, by generalizing a
result of L\"{u}cking et al.~\cite[Theorem~5.4]{luecking}.
This implies that the benefits of higher altruism in singleton
congestion games are only reaped in pure Nash equilibria,
and the gap between the pure and mixed price of anarchy increases in $\alpha$.
Also it shows that singleton congestion games constitute a class of games for which the smoothness argument cannot deliver tight bounds.

\begin{proposition}\label{prop:mixedpoalb}
For every $\alpha \in [0,1]^n$, the mixed price of anarchy for
$\alpha$-altruistic singleton linear congestion games is at least $2$.
\end{proposition}
\begin{proof}
Let $m \geq 2$ and consider the instance with player set
$\{1, \ldots, m\}$ and facility set $\{1, \ldots, m\}$,
with $d_e(x) = x$ ($a_e = 1$ and $b_e = 0$) for each facility $e$.
Denote by $s$ the mixed strategy where each player chooses each link
with probability $1/m$. When $\alpha_i = 0$ for every player, $s$
is a mixed Nash equilibrium, and $\mathbf{E}[C(s)] = 2m - 1$ as proved
in \cite{luecking}. The optimum is clearly $m$, so the price of
anarchy of this instance is $2 - 1/m$.

All that is left to show is that $s$ is also a Nash equilibrium under
arbitrary altruism levels.
By symmetry, it suffices to show that the expected cost of player $1$
increases if he deviates to the strategy where he chooses facility $1$
with probability 1. Let $s_1^* = 1$. We have
\begin{eqnarray*}
\mathbf{E}[C_1^{\alpha}(s_1^*, s_{-1})] & = & \mathbf{E}[(1-\alpha_1) C_1(s_1^*, s_{-1}) + \alpha_1 C(s_1^*, s_{-1})] \\
 & = & (1-\alpha_1) \mathbf{E}[C_1(s_1^*, s_{-1})] + \alpha_1 \mathbf{E}[C(s_1^*, s_{-1})] .
\end{eqnarray*}
We already know that $\mathbf{E}[C_1(s_1^*, s_{-1})] \geq
\mathbf{E}[C_1(s)]$ because $s$ is a Nash equilibrium when the players
are completely selfish, so we are done when we show
$\mathbf{E}[C(s_1^*, s_{-1})] \geq \mathbf{E}[C(s)] = 2m - 1$.

For an arbitrary pure strategy profile $s$, let $X_{i,e}(s')$ be the
indicator function that maps to $1$ if player $i$ chooses facility $e$
under $s'$, and $0$ otherwise. Then it is clear that
$C_i(s') = \sum_{e=1}^m X_{i,e}(s')d_e(s')$ for
$i = 1, \ldots, m$, and
$d_e(s) = \sum_{i=1}^m X_{i,e}(s')$ for $e = 1, \ldots, m$.
So $C_i(s') = \sum_{e,j=1}^m X_{i,e}(s')X_{j,e}(s')$.
Using this last identity, along with symmetry, independence, and
linearity of  expectation, the following derivation is easily made
(letting $s' = (s_1^*, s_{-1}$)):
\begin{eqnarray*}
\mathbf{E}[C(s_1^*, s_{-1})] & = & \sum_{i = 1}^m \mathbf{E}[C_i(s')] \\
& = & \mathbf{E}[C_1(s')] + (m-1)\mathbf{E}[C_2(s')] \\
& = & \mathbf{E}[d_1(s')] + (m-1)\sum_{e,j=1}^m \mathbf{E}[X_{2,e}(s')X_{j,e}(s')] \\
& = & \sum_{i=1}^m \mathbf{E}[X_{i,1}(s')] + (m-1)\left(\sum_{j=1}^m \mathbf{E}[X_{2,1}(s')X_{j,1}(s')]
+ (m-1)\sum_{j = 1}^m \mathbf{E}[X_{2,2}(s')X_{j,2}(s')]\right) \\
& = & \left(1 + (m-1)\frac{1}{m}\right) + (m-1)\left(\frac{1}{m} + \frac{1}{m} + (m-2)\frac{1}{m^2}
+ (m-1)\left(0 + \frac{1}{m} + (m-2)\frac{1}{m^2}\right)\right) \\
& = & 2m -1.
\end{eqnarray*}
\end{proof}

\subsection{Non-Uniform Altruism}

%\gsdelete{As a first step to extend the analysis to non-uniform altruism,}
We analyze
the case when all altruism levels are in $\SET{0,1}$, i.e., each
player is either completely altruistic or completely
selfish.\footnote{This model relates naturally to \emph{Stackelberg
    scheduling games} (see, e.g., \cite{chenkempe}).}
Then, the system is entirely characterized by the fraction $\alpha$ of
altruistic players (which coincides with the average altruism level).
The next theorem shows that in this case, too, the pure price of anarchy \emph{improves} with the overall altruism level.

\begin{theorem}\label{thm:0-1bound}
Assume that an $\alpha$ fraction of the players are completely
altruistic, and the remaining $(1-\alpha)$ fraction are completely
selfish.
Then, the pure price of anarchy of the altruistic singleton linear congestion game
is at most $\frac{4-2\alpha}{3-\alpha}$.
\end{theorem}

Let $s$ be a pure Nash equilibrium of
$G^\altvec$ and $s^*$ an optimal strategy profile.
Again, let $x_e = x_e(s)$ and $x^*_e = x_e(s^*)$.
Based on the strategy profile $s$, we partition the edges in $E$ into
sets $E_0, E_1$:
\begin{equation*}
E_1 = \{ e \in E: \text{$\exists i \in N$ with $\alpha_i = 1$ and $s_i = \{e\}$}\} \enspace ,
\end{equation*}
is the set of edges having at least one altruistic player, while
$E_0 = E \setminus E_1$ is the set of edges that are used exclusively by selfish
players or not used at all. Let $N_1$ and $N_0$ refer to the respective player sets that
are assigned to $E_1$ and $E_0$.
$N_1$ may contain both altruistic and selfish players, while $N_0$
consists of selfish players only.
Let $k _1= \sum_{e \in E_1} x_e$ and $k_0 = n-k_1$ denote the number
of players in $N_1$ and $N_0$, respectively.

The high-level approach of our proof is as follows: We split the total
cost $C(s)$ of the pure Nash equilibrium into $C(s) = \gamma C(s) +
(1-\gamma) C(s)$ for some $\gamma \in [0,1]$ such that $\gamma C(s) =
\sum_{e \in E_0} x_e d_e(x_e)$ and $(1-\gamma)C(s) = \sum_{e \in E_1}
x_e d_e(x_e)$. We bound these two contributions separately to
show that
\begin{equation}\label{eq:comb}
{\textstyle\frac34} \gamma C(s) + (1-\gamma)C(s) \le C(s^*) \enspace .
\end{equation}
The pure price of anarchy is therefore at most $(\frac34 \gamma +
(1-\gamma))^{-1} = \frac{4}{4-\gamma}$. The bound then follows by deriving an upper bound on
$\gamma$ in Lemma \ref{lem:gamma}.

\begin{lemma}\label{lem:dom}
Let $s$ be a pure Nash equilibrium and assume that the delay functions $(d_e)_{e \in E}$ are
semi-convex. Then there is an optimal strategy profile $s^*$ such that
$x_e(s) \le x_e(s^*)$ for every edge $e \in E_1$.
\end{lemma}
\begin{proof}
Let $s^*$ be an optimal strategy profile, let $x_e$ denote $x_e(s)$, let $x_e^*$ denote $x_e(s^*)$, and assume that $x^*_e < x_e$ for
some $e \in E_1$. Then there is some edge $\bar{e} \in E$ with
$x^*_{\bar e} > x_{\bar e}$. Consider an altruistic player $i \in N_1$
with $s_i = \{e\}$. (Note that $i$ must exist by the definition of
$E_1$.) Because $s$ is a pure Nash equilibrium, player $i$ has no
incentive to deviate from $e$ to $\bar{e}$, i.e., $C(\{\bar{e}\},
s_{-i}) \ge C(s)$, or, equivalently,
\begin{equation}\label{eq:qq1}
(x_{\bar e}+1) d_{\bar e}(x_{\bar e}+1) - x_{\bar e} d_{\bar e}(x_{\bar e}) \geq x_{e} d_e(x_{e}) - (x_{e} -1) d_{e}(x_{e}-1) \enspace.
\end{equation}
Since $x^*_e < x_e$ and $x_{\bar e} < x^*_{\bar e}$, the semi-convexity
of the delay functions implies
\begin{eqnarray}
(x^*_e +1) d_e(x^*_e +1) - x^*_e d_e(x^*_e) & \leq & x_{e} d_e(x_{e}) - (x_{e} -1) d_{e}(x_{e}-1) \enspace, \label{eq:qq2} \\
(x_{\bar e}+1) d_{\bar e}(x_{\bar e}+1) - x_{\bar e} d_{\bar e}(x_{\bar e}) & \leq & x^*_{\bar e} d_{\bar e}(x^*_{\bar e}) - (x^*_{\bar e} -1) d_{\bar e}(x^*_{\bar e}-1) \enspace. \label{eq:qq3}
\end{eqnarray}
By combining \eqref{eq:qq1}, \eqref{eq:qq2} and \eqref{eq:qq3} and re-arranging terms, we obtain
\begin{equation*}
(x^*_e +1) d_e(x^*_e +1) + (x^*_{\bar e} -1) d_{\bar e}(x^*_{\bar e}-1) \leq x^*_e d_e(x^*_e) + x^*_{\bar e} d_{\bar e}(x^*_{\bar e}) \enspace .
\end{equation*}
The above inequality implies that by moving a player $j$ with $s^*_j = \{ \bar e\}$ from $\bar e$ to $e$, we obtain a new strategy profile $s' = (\{ e\}, s^*_{-j})$ of cost $C(s') \le C(s^*)$. (Note that $j$ must exist because $x^*_{\bar e} > x_{\bar e} \ge 0$.) Moreover, the number of players on $e$ under the new strategy profile $s'$ increased by one. We can therefore repeat the above argument (with $s'$ in place of $s^*$) until we obtain an optimal strategy profile that satisfies the claim.
%\qed
\end{proof}

Note that Lemma~\ref{lem:dom} implies that at least for singleton congestion games, entirely altruistic players will ensure that Nash equilibria are optimal.

\begin{corollary}\label{cor:semi-convex}
The pure price of anarchy of $1$-altruistic extensions of symmetric
singleton congestion games with semi-convex delay functions is $1$.
\end{corollary}

Henceforth, we assume that $s^*$ is an optimal strategy profile that satisfies the statement of Lemma~\ref{lem:dom}.

%The high-level idea to prove Theorem~\ref{thm:0-1bound} is as follows:

\begin{lemma}\label{lem:decomp}
Define $y^*$ as $y^*_e = x^*_e - x_e \ge 0$ for every $e \in E_1$, and $y^*_e = x^*_e$ for all edges $e \in E_0$.
Then, $\sum_{e \in E_0} x_e d_e(x_e) \le \textstyle\frac43
\sum_{e \in E} y^*_e d_e(x^*_e)$.
\end{lemma}
\begin{proof}
Consider the game $\bar G$ induced by $G^\altvec$ if all $k_1$ players
in $N_1$ are fixed on the edges in $E_1$ according to $s$. Note that
all remaining $k_0 = n - k_1$ players in $N_0$ are selfish. That is,
$\bar G$ is a symmetric singleton congestion game with player set
$N_0$, edge set $E$ and delay functions $(\bar d_e)_{e \in E}$, where
$\bar{d}_e(z) = d_e(x_e + z)$ if $e \in E_1$ and $\bar d_e(z) =
d_e(z)$ for $e \in E_0$. Let $\bar s$ be the restriction of $s$ to the
players in $N_0$, and define $\bar x$ as $\bar x_e = 0$ for $e \in E_1$ and $\bar x_e = x_e$ for $e \in E_0$.
It is not hard to verify that $\bar s$ is a pure Nash equilibrium of
the game $\bar G$.
Let $\bar{s}^*$ be a socially optimum profile for $\bar{G}$, and
for each edge $e$, let $\bar{x}^*_e$ be the total number of players on
$e$ under $\bar{S}^*$. Then,

\[
\sum_{e \in E_0} x_e d_e(x_e)
= \sum_{e \in E} \bar{x}_e \bar d_e(\bar x_e) %\\
\leq \frac43 \sum_{e \in E} \bar x^*_e \bar d_e(\bar x^*_e) %\\
\leq \frac43 \sum_{e \in E} y^*_e \bar{d}_e(y^*_e) %\\
= \frac43 \sum_{e \in E} y^*_e d_e(x^*_e) \enspace ,
\]
where the first inequality follows from Theorem~\ref{thm:poa-sscg} and the second inequality follows from the optimality of $\bar x^*$.
%\qed
\end{proof}

\begin{lemma}\label{lem:gamma} It holds that $\gamma \leq \frac{2n_0}{n+n_0}=\frac{2(1-\alpha)}{2-\alpha}$.
\end{lemma}
\begin{proof}
The claim follows directly from Theorem~\ref{thm:poa-sscg} if $N_1 = \emptyset$.
Assume that $N_1 \neq \emptyset$, and let $j \in N_1$ with
$s_j = \{\bar e\}$.
Let $\bar C(s) = \sum_{i \in N_0} C_i(s)/k_0$ be the average
cost experienced by players in $N_0$. We first show $C_j(s) \geq
\frac12 \bar C(s)$. If $N_0 = \emptyset$, then $C_j(s) \geq \frac12
\bar C(s)$ trivially holds. Suppose that $N_0 \neq \emptyset$, and let
$i \in N_0$ with $s_i = \{e\}$. Recall that $i$ is selfish. Because $s$
is a Nash equilibrium, we have
\[ C_i(s) = a_e x_e + b_e \le a_{\bar e} (x_{\bar e} + 1) + b_{\bar e}
\le 2(a_{\bar e} x_{\bar e}  + b_{\bar e}) = 2 C_j(s).
\]
By summing over all $k_0$ selfish players in $N_0$,
we obtain $C_j(s) \ge \frac12 \bar C(s)$ and thus
$\sum_{j\in N_1} C_{j}(s) \ge \frac12 k_1 \bar{C}(S)$.
We have
\[%\begin{eqnarray*}
\gamma = \frac{\sum_{i\in N_0} C_i(s)}{\sum_{i\in N_0} C_i(s) + \sum_{j\in N_1} C_{j}(s)} %\\
 \leq \frac{k_0\bar{C}(S)}{k_0\bar{C}(S)+\frac12 k_1 \bar{C}(S)} %\\
= \frac{2k_0}{n+k_0} %\\
\leq \frac{2 n_0}{n+n_0} \enspace ,
\]%\end{eqnarray*}
where the last inequality follows because $k_0 \le n_0$.
%\qed
\end{proof}

%We can finally prove Theorem~\ref{thm:0-1bound}.
\begin{proof}[Proof of Theorem~\ref{thm:0-1bound}]
Using the above lemmas, we can show that the relation in \eqref{eq:comb} holds:
\begin{eqnarray*}
\frac34 \gamma C(s) + (1-\gamma)C(s) & = & \frac34 \sum_{e \in E_0}  x_e d_e(x_e) + \sum_{e \in E_1} x_e d_e(x_e)
 \leq  \sum_{e \in E}  y^*_e d_e(x^*_e) + \sum_{e \in E_1} x_e d_e(x_e) \\
& = & \sum_{e \in E}  x^*_e d_e(x^*_e) + \sum_{e \in E_1} (x_e d_e(x_e) - x_e d_e(x^*_e))
\leq \sum_{e \in E}  x^*_e d_e(x^*_e) = C(s^*) \enspace ,
\end{eqnarray*}
where the first inequality follows from Lemma~\ref{lem:decomp} and the
last inequality follows from Lemma~\ref{lem:dom} and because delay
functions are monotone increasing. We conclude that the pure price of
anarchy is at most
\[
\left(\frac34 \gamma + (1-\gamma)\right)^{-1}
= \frac{4}{4-\gamma}\leq\frac{4-2\alpha}{3-\alpha}.
\]
The bound now follows from Lemma~\ref{lem:gamma}.
% \qed
\end{proof}

\section{General Properties of Smoothness}\label{sec:general}
For the game classes that we analyzed (with the exception of symmetric singleton congestion
games), we used $(\lambda,\mu,\altvec)$-smoothness as our main tool to
derive bounds on the price of anarchy. In this section, we provide
some general results about $(\lambda, \mu, \altvec)$-smoothness.

\begin{proposition}\label{prop:convex}
Suppose that $\mathcal{G}$ is a class of cost-minimization games
equipped with sum-bounded social cost functions. The set
$S_{\mathcal{G}} = \{(\lambda, \mu, \altvec) : \forall G \in
\mathcal{G} $, $G^\alpha$ is $(\lambda, \mu, \altvec)$-smooth$\}$ is
convex.
\end{proposition}
\begin{proof}
Pick an arbitrary game $G \in \mathcal{G}$. It suffices to show that $S_G = \{(\lambda, \mu, \altvec) : \text{$G^\alpha$ is $(\lambda, \mu, \altvec)$-smooth}\}$ is convex, because the intersection of any collection of convex sets is always convex.

Let $(\lambda_1, \mu_1, \altvec^1), (\lambda_2, \mu_2, \altvec^2) \in S_G$ be two elements in $S_G$, and pick an arbitrary $\gamma \in [0,1]$. For all pairs $(s, s^*)$ of strategy profiles of $G$,
\begin{eqnarray*}
& & \gamma \sum_{i = 1}^n (C_i(s_i^*, s_{-i}) + \alpha_i^1(C_{-i}(s_i^*, s_{-i}) - C_{-i}(s)))
+\ (1-\gamma)\sum_{i = 1}^n (C_i(s_i^*, s_{-i}) + \alpha_i^2(C_{-i}(s_i^*, s_{-i}) - C_{-i}(s))) \\
 & \leq & \gamma (\lambda_1 C(s^*) + \mu_1 C(s)) + (1 - \gamma)(\lambda_2 C(s^*) + \mu_2 C(s)) \enspace .
\end{eqnarray*}
By rewriting both sides of the above inequality, we obtain
%We rewrite the left hand side by placing the $\gamma$ and $(1 - \gamma)$ inside of the summations:
\begin{align*}
& \sum_{i = 1}^n (C_i(s_i^*, s_{-i}) + (\gamma\alpha_i^1 + (1 - \gamma)\alpha_i^2)(C_{-i}(s_i^*, s_{-i}) - C_{-i}(s))) \\
& \qquad \leq (\gamma \lambda_1 + (1 - \gamma) \lambda_2) C(s^*) + (\gamma \mu_1 + (1 - \gamma)\mu_2) C(s) \enspace.
\end{align*}
We conclude that $G$ is $(\gamma(\lambda_1, \mu_1, \altvec^1) + (1 -
\gamma)(\lambda_2, \mu_2, \altvec^2))$-smooth. Therefore, $S_G$ is
convex.
%\qed
\end{proof}

A natural question to ask is whether the robust price of anarchy is also a
convex function of $\altvec$. This turns out not to be the
case.
For instance, the robust price of anarchy for uniformly
$\alpha$-altruistic congestion games is
$\frac{5 + 4\alpha}{2 + \alpha}$ (see Section~\ref{sec:cg}),
which is a non-convex function. However, we can
prove a somewhat weaker statement: For a subset $S \subseteq
\mathbb{R}^n$, we call a function $f : S \rightarrow \mathbb{R}$
\emph{quasi-convex} iff
$f(\gamma x + (1 - \gamma) y) \leq \max\{f(x), f(y)\}$
for all $\gamma \in [0,1]$.

\begin{theorem}\label{thm:rpoa-quasiconvex}
Let $\mathcal{G}$ be a class of games equipped with sum-bounded social cost functions. Then $\RPoA_{\mathcal{G}}(\altvec)$ is a quasi-convex function of $\altvec$.
\end{theorem}
\begin{proof}
Let $G \in \mathcal{G}$. We show that for any $\altvec^1, \altvec^2 \in \mathbb{R}^n$ and $\gamma \in [0,1]$,
\begin{equation*}
\RPoA(\gamma\altvec^1 + (1 - \gamma)\altvec^2) \leq \max\{\RPoA(\altvec^1), \RPoA(\altvec^2)\} \enspace .
\end{equation*}
Let $(\epsilon_1, \epsilon_2, \ldots)$ be a decreasing sequence of positive real numbers that tends to $0$. Moreover, let
\begin{equation*}
((\lambda_{1,1}, \mu_{1,1}, \altvec^1), (\lambda_{1,2}, \mu_{1,2}, \altvec^1), \ldots) \quad\text{and}\quad ((\lambda_{2,1}, \mu_{2,1}, \altvec^2), (\lambda_{2,2}, \mu_{2,2}, \altvec^2), \ldots)
\end{equation*}
be sequences of elements in $S_G$ (where $S_G$ is as defined in the proof of Proposition \ref{prop:convex}) such that
\begin{equation*}
\RPoA(\altvec^1) + \epsilon_j = \textstyle\frac{\lambda_{1,j}}{1 -
  \mu_{1,j}} \quad \text{and} \quad \RPoA(\altvec^2) +
\epsilon_j = \textstyle\frac{\lambda_{2,j}}{1 - \mu_{2,j}}
\end{equation*}
for all $j$. By Proposition \ref{prop:convex}, we know that for all $j$,
\begin{align*}
 & \sum_{i = 1}^n (C_i(s_i^*, s_{-i}) + (\gamma\alpha_i^1 + (1 - \gamma)\alpha_i^2)(C_{-i}(s_i^*, s_{-i}) - C_{-i}(s))) \\
& \qquad \leq \gamma (\lambda_{1,j} C(s^*) + \mu_{1,j} C(s)) + (1 - \gamma)(\lambda_{2,j} C(s^*) + \mu_{2,j} C(s)) \\
 & \qquad \leq \max\{\lambda_{1,j} C(s^*) + \mu_{1,j} C(s), \lambda_{2,j} C(s^*) + \mu_{2,j} C(s)\} \enspace .
\end{align*}
Hence,
\[
\RPoA(\gamma \altvec^1 + (1 - \gamma) \altvec^2) \leq \max\left\{\frac{\lambda_{1,j}}{1 - \mu_{1,j}} , \frac{\lambda_{2,j}}{1 - \mu_{2,j}} \right\}
 \leq \max\{\RPoA(\altvec^1), \RPoA(\altvec^2)\} + \epsilon_j \enspace ,
\]
for all $j$. By taking the limit as $j$ goes to infinity, we conclude
$\RPoA(\gamma\altvec^1 + (1 - \gamma)\altvec^2) \leq \max\{\RPoA(\altvec^1), \RPoA(\altvec^2)\}$, which proves
the claim.
%\qed
\end{proof}

The quasi-convexity of $\RPoA_{\mathcal{G}}$ implies:
\begin{corollary}
The points $\altvec$ that minimize $\RPoA_{\mathcal{G}}(\altvec)$ on
the domain $[0,1]^n$ form a convex set. The set of points $\altvec$ that maximize
$\RPoA_{\mathcal{G}}(\altvec)$ on the domain $[0,1]^n$ includes at
least one point that is a 0-1 vector.
\end{corollary}

\section{Conclusions and Future Work}\label{sec:discussion}

One might not expect that there are games in which the price of anarchy is greater than $1$ when $\altvec = \mathbf{1}$. This phenomenon is a lot less surprising when approached from a local search point-of-view, as this is only equivalent to saying that there exist local optima in the objective function $C$ with respect to the neighborhood set obtained by taking all strategies obtained by single-player deviations from a given strategy profile $s$. Nevertheless, it still seems to us rather surprising that the price of anarchy can get worse when the altruism level $\altvec$ gets closer to $\mathbf{1}$. This phenomenon has been observed before, in \cite{caragiannis}.
The fact that the price of anarchy does not \textit{necessarily} get worse in all cases is exemplified by our analysis of the pure price of anarchy in symmetric singleton congestion games.

The most immediate future directions include analyzing singleton
congestion games with more general delay functions than linear ones.
While the price of anarchy of such functions increases (e.g., the price of anarchy for
polynomials increases exponentially in the degree \cite{AAE05,christodoulou}),
this also creates room for potentially larger reductions due to altruism.
Similarly, the characterization of the robust price of anarchy of altruistic congestion
games with more general delay functions (e.g., polynomials) is left for future work.

For games where the smoothness argument cannot give tight
bounds, would a refined smoothness argument like local smoothness in
\cite{roughgarden3} work? For symmetric singleton congestion games,
this seems unlikely, as the price of anarchy bounds are already different between
pure and mixed Nash equilibria.
It is also worth trying to apply the smoothness argument
or its refinements to analyze the price of anarchy for other dynamics in other classes
of altruistic games, for example, (altruistic) network vaccination
games \cite{chen:david}, which are known to not always possess pure
Nash equilibria, or to find examples to see why smoothness-based arguments do not work.

We have seen that the impact of altruism depends on the underlying game.
It would be nice to identify general properties that enable to predict whether a given game suffers from altruism or not. What is it that makes valid utility game invariant to altruism?
Furthermore, what kind of ``transformations" (not just altruistic extensions) might be applied to a strategic game such that the smoothness approach can still be adapted to give (tight) bounds?
More generally, while the existence of pure Nash equilibria has been shown
for singleton and matroid congestion games with player-specific
latency functions \cite{ackermann,milchtaich}, the price of anarchy (for pure Nash
equilibria or more general equilibrium concepts) has not yet been addressed.
Studying the price of anarchy in such a general setting (in which our setting with
altruism can be embedded) by either smoothness-based techniques or
other methods is undoubtedly intriguing.

\newpage

%\begin{small}
\bibliographystyle{plain}
\bibliography{smoothness-full}
%\end{small}

\end{document}